\documentclass[english,a4paper]{article}

\usepackage{graphicx}
\usepackage{amsmath}
\usepackage{amssymb}
\usepackage{mathtools}
\usepackage{color}

\def\phi{\varphi}
\def\epsilon{\varepsilon}

\everymath{\displaystyle}

\newtheorem{theorem}{Theorem}
\newtheorem{prop}[theorem]{Proposition}

\newtheorem{lemma}[theorem]{Lemma}

\newcommand{\RR}{\mathbb{R}}
\newcommand{\CC}{\mathbb{C}}

\begin{document}
\title{Hamiltonian Monodromy via spectral Lax pairs}
\author{G. J. Gutierrez Guillen\footnote{Institut de Math\'ematiques de Bourgogne - UMR 5584 CNRS, Universit\'e de Bourgogne-Franche-Comt\'e,
	9 avenue Alain Savary, BP 47870, 21078 DIJON, FRANCE}, D. Sugny\footnote{Laboratoire Interdisciplinaire Carnot de Bourgogne (ICB), UMR 6303 CNRS-Universit\'e Bourgogne-Franche Comt\'e, 9 Av. A.
Savary, BP 47 870, F-21078 Dijon Cedex, France; dominique.sugny@u-bourgogne.fr}, P. Marde\v si\'c\footnote{Institut de Math\'ematiques de Bourgogne - UMR 5584 CNRS, Universit\'e de Bourgogne-Franche-Comt\'e,
	9 avenue Alain Savary, BP 47870, 21078 DIJON, FRANCE; pavao.mardesic@u-bourgogne.fr}}

\maketitle

\begin{abstract}
Hamiltonian Monodromy is the simplest topological obstruction to the existence of global action-angle coordinates in a completely integrable system. We show that this property can be studied in a neighborhood of a focus-focus singularity by a spectral Lax pair approach. From the Lax pair, we derive a Riemann surface which allows us to compute in a straightforward way the corresponding Monodromy matrix. The general results are applied to the Jaynes-Cummings model and the spherical pendulum.
\end{abstract}


\section{Introduction}

Hamiltonian integrable systems with a finite number of degrees of freedom have a long history going back from Liouville in the mid-nineteenth century to Arnold one hundred years later~\cite{anorldbook}. A modern description of integrable systems was formulated in the last decades in terms of Lax pairs~\cite{lax1968,babelon}. In this context, a Lax pair consists of two matrix-valued functions on the phase space satisfying a differential equation equivalent to the Hamiltonian dynamics. When one is able to derive such a Lax pair, this approach is a powerful tool to find the constants of motion of the integrable system. On the other side, geometric or topological properties are known to provide valuable insights on the dynamics and the structure of such complex systems~\cite{anorldbook,bolsinovbook}. The richness of this approach for integrable systems is illustrated by the concept of Hamiltonian Monodromy (HM) which was introduced by Duistermaat in 1980~\cite{duistermaat}. HM is the simplest topological obstruction to the existence of global action-angle coordinates in a completely integrable Hamiltonian system. In a two-degree of freedom system, a non-trivial Monodromy can be observed if the set of regular values of the energy-momentum map is not simply connected but the image of this mapping, called the bifurcation diagram, has an isolated singular point corresponding to a focus-focus singularity~\cite{zung1997}. Non-trivial HM has been found in a wide variety of mechanical systems such as the spherical pendulum or the champagne bottle~\cite{cushmannbook,efstathioubook,efstathiou2017b,bates1991}. This concept has been generalized to nonstandard forms of Monodromy extending from fractional Monodromy~\cite{nekhoroshev2006,giacobbe2008,efstathiou2013} and bidromy~\cite{sadovskii2007,efstathiou2010} which characterize a line of non isolated weak singularities of the bifurcation diagram, to scattering Monodromy for non compact dynamics~\cite{bates2007,dullin2008} and dynamical Monodromy, a property that can be observed with nonautonomous systems~\cite{delos2009,delos2018}. We refer the reader to recent review papers for details about this active research topic~\cite{sadovskii2016,efstathiou2021}. The quantum analogue of HM was formulated mathematically few years later~\cite{cusman1988,san1999,guillemin1989} and was also at the origin of many studies in physics, as a way to describe the global structure of quantum spectra~\cite{child07,cushman2000,cushman2004,dullin2018,assemat2010,babelon2009,pelayo2012,sadovskii1999}.

For integrable systems with compact orbits, the phase space is fibered by tori or by disjoint union of tori over the regular values of the bifurcation diagram. HM describes the possible nontriviality of the torus bundle over a loop in the set of such regular values~\cite{cushmannbook,efstathioubook}. HM can be viewed as the holonomy of a connection of the bundle which allows one to transport a basis of the first homology group of the torus. HM is characterized by a matrix with integer coefficients, the Monodromy matrix, that corresponds to the transformation of this basis along the loop. In this study, we consider two degree of freedom system with a global circle action over the phase space~\cite{cushmannbook,efstathioubook}. In this case, HM can be studied by analyzing a function on the phase space, namely the rotation number $\Theta$. For a focus-focus singularity, the variation of $\Theta$ along the loop is equal to $2\pi$ and completely characterizes the nontrivial Monodromy of the corresponding Hamiltonian system. The computation of this function and its variation are therefore the building blocks for describing HM~\cite{cushmannbook,efstathioubook,efstathiou2017}. Note that more geometric approaches based on the global structure of the torus bundle can also be used to show the non-trivial Monodromy of a Hamiltonian system~\cite{efstathiou2020}.

A similar concept of Monodromy appears in complex geometry with the Picard-Lefschetz theory for Riemann surfaces~\cite{arnoldbook2,zoladekbook}. A natural question is to determine what relation may exist between these two forms of Monodromy. A first answer is given by scattering Monodromy in a hyperbolic oscillator for which a direct link can be established in a complex system with a $A_1$ singularity. Complex geometry can also be used as an efficient way to compute the Monodromy matrix. Analytical computations are easier in this approach and generally boil down to residue calculus, while a major difficulty is to keep track of the real nature of the problem.

This method can be applied in different forms. The first option consists in complexifying directly the real system and then studying the function $\Theta$, which can be expressed as an Abelian integral. The variation of $\Theta$ along a loop in the bifurcation diagram can be achieved from the Picard-Lefschetz Monodromy of the corresponding cycle. This strategy has, for instance, be used in the spherical pendulum in~\cite{beukers} and in resonant systems in~\cite{sugny2008}. A second approach is based on a Lax pair description, which leads naturally to a Riemann surface encoding all the dynamical information of the Hamiltonian system. Different attempts have been proposed in this direction for a series of examples ranging from the spherical pendulum~\cite{audin2002,gavrilov2002} to the Lagrange top~\cite{gavrilov1998,vivolo} and the Jaynes-Cummings model (or spin-oscillator system)~\cite{babelon2012,babelon2015}. These systems have the characteristics of having both a non-trivial Monodromy and known Lax pairs. However, some works about the spherical pendulum~\cite{audin2002,gavrilov2002} do not establish a direct link with the system dynamics and use sophisticated mathematical objects which are difficult to connect with standard proofs of HM. Other studies that focus on the JC model~\cite{babelon2012,babelon2015} do not compute the Monodromy matrix.

We propose in this paper to analyze this procedure in a direct and explicit way by showing in particular how to compute the variation of the rotation number from the Lax pair formalism. More precisely, starting from the Lax pair of the Hamiltonian system under study, we show that a complex reduced phase space can be defined as a set of Riemann surfaces depending on the constants of motion of the system. On the basis of the various known examples, we establish generic conditions that this complex fibration must satisfy in order to exhibit a nontrivial Monodromy. We derive a normal form for this family of Riemann surfaces to compute explicitly the Picard-Lefschetz Monodromy. In this setting, we show that the rotation number is expressed as an Abelian integral of a meromorphic one-form with a nonzero residue at infinity. The last step of our approach consists in combining the two results to get the Monodromy matrix. We then apply this general approach to the Jaynes-Cummings model and to the spherical pendulum. Since a non-trivial monodromy is not limited to systems for which a Lax pair is known, we introduce, what we call a quasi-Lax pair formalism. It corresponds to a Lax pair up to higher order terms valid in a neighborhood of focus-focus points. It can be derived for any Hamiltonian system locally around this singularity. Using the same procedure as above with some adaptations, we obtain the Monodromy matrix characterizing the focus-focus singularity. Some open questions and generalizations to higher-dimensional integrable systems are discussed in a final section.

The paper is organized as follows. The methodology and the main results are presented in Sec.~\ref{sec.metho}. Section~\ref{sec.tech} is dedicated to the analysis of the properties of the set of Riemann surfaces describing a Hamiltonian system with a focus-focus singularity. The proofs of the different results are presented in Sec.~\ref{sec.proof}. The general results of this study are applied to two examples in Sec.~\ref{sec.exam}, namely the Jaynes-Cummings model and the spherical pendulum. The quasi-Lax pair is introduced in Sec.~\ref{secquasi}. Conclusion and prospectives are given in Sec.~\ref{sec.conclu}. Complementary computations are reported in Appendices~\ref{appA} and \ref{appB}.

\section{Methodology and main results}\label{sec.metho}
We study completely integrable Hamiltonian systems in $\mathbb{R}^4$. Recall that the \textit{energy momentum map} is defined as
\begin{eqnarray*}
\mathcal{EM}\colon & & \RR^{4}\to\RR^2\\
& &(q,p)\mapsto (H(q,p),K(q,p)) ,
\end{eqnarray*}
where $H$ and $K$ are two independent constants of motion. This map defines a fibration and we assume that the fibers are compact and connected. \textit{The bifurcation diagram} is the image of the energy momentum map and \textit{the bifurcation set} of this mapping is the set in $\RR^2$ over which $\mathcal{EM}$ fails to be a locally trivial fibration. The Arnold-Liouville theorem gives the existence of \textit{action-angle} coordinates locally when $\mathcal{EM}$ is proper (outside of the bifurcation set). HM is an obstruction to the existence of global action-angle coordinates. In this paper, we study Hamiltonian Monodromy along loops around an isolated critical value of $\mathcal{EM}$, via spectral Lax pairs. More specifically, we derive a Riemann surface using the spectral Lax pair. Then, in Theorem~\ref{Teo.main}, we formulate a general result about integrals of meromorphic one-forms on this Riemann surface and, in Theorem~\ref{cor.main}, we give conditions under which the Monodromy is not trivial and calculate it using Theorem~\ref{Teo.main}. These conditions are obtained by studying well-known physical systems and their properties. Additionally, we present in Sec.~\ref{sec.exam} some of these systems (which have a focus-focus singularity) and show that Theorem~\ref{cor.main} applies.

Given a completely integrable system, a \textit{Lax pair} describing this system is a pair of square matrices $L$ and $M$ which fulfill the equation
\begin{equation*}
\dot{L}=[M,L] ,
\end{equation*}
and such that, these differential equations are equivalent to the equations of motion of the system. A \textit{spectral Lax pair} is a Lax pair for which the matrices $L(\lambda)$ and $M(\lambda)$ depend on an extra parameter $\lambda$ called \textit{spectral parameter}. If the dependence is polynomial then the spectral Lax pair defines a \textit{spectral curve} by $\det(L(\lambda)-\mu I )=0$, which is a polynomial in the variables $\lambda$ and $\mu$ such that all its coefficients are constants of motion.

We describe now the properties of the systems we consider. We focus on the study of Hamiltonian systems in $\RR^4$ with Hamiltonian $H$ and an $\mathbb{S}^1$-action generated by the second constant of motion $K$. We denote respectively by $h$, $k$ the values of $H$ and $K$. Since $K$ generates an $\mathbb{S}^1$-action, its flow corresponds to closed orbits of period $2\pi$. A reduced phase space can be defined over the regular values of the bifurcation diagram as the set of the orbits of $K$. Consider a point $(q,p)$ belonging to a regular torus $\mathcal{EM}^{-1}(h,k)$, and the closed orbit of $K$ starting at this point. Let $\gamma$ be the orbit of the flow generated by $H$ starting at $(q,p)$ and ending at the first point of intersection with the orbit of $K$ at time $T$ called the first return time (see Fig.~\ref{fig.gamma}). The image of $\gamma$ in the reduced phase space is a cycle $\delta$, which is represented in Fig.~\ref{fig.gamma}.

\begin{figure}[h!]
  \centering
    \includegraphics[scale=.2]{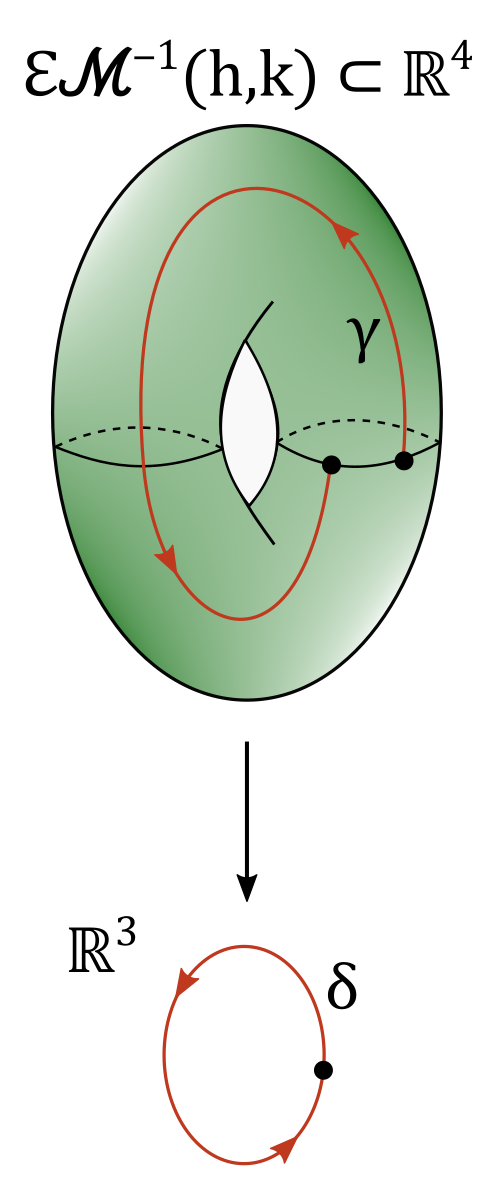}
    \caption{Projection $\delta$ of the Hamiltonian flow $\gamma$ onto the reduced phase space. The horizontal black circles on the torus represent the orbits of $K$. The black dots correspond to the intersection points between $\gamma$ and one orbit of $K$.}
    \label{fig.gamma}
\end{figure}

If we denote by $\theta$ the canonically conjugate angle to $K$ then the rotation number $\Theta$ can be defined as
\begin{equation*}
\Theta=\int_{\gamma}d\theta=\int_{0}^T\dot{\theta}dt .
\end{equation*}

We now assume that the system can be described by a spectral Lax pair, $L(\lambda)$, $M(\lambda)$, where $L$ and $M$ are complex $2\times 2$ matrices. Suppose that the matrix $L$ can be expressed as follows
\begin{equation}\label{eq.matL}
L(\lambda)=\begin{pmatrix}
A(\lambda) & B(\lambda) \\
C(\lambda) & -A(\lambda)
\end{pmatrix} ,
\end{equation}
where $A$, $B$ and $C$ are polynomial functions of the parameter $\lambda$, with $A$ of degree two, $B$ and $C$ of degree one. The spectral curve can be derived from the characteristic polynomial of $L$ and is given by the equation $\mu^2=A^2(\lambda)+B(\lambda)C(\lambda)=Q_{h,k}(\lambda)$, where $\mu$ denotes the eigenvalue of $L$ and $Q$ is a polynomial of degree four. Using the standard procedure of Lax pairs for defining separated variables~\cite{babelon}, we then introduce the functions $\tilde{\lambda}$ and $\tilde{\mu}$ respectively as the solution of the implicit equation $C(\lambda)=0$ and by the relation $\tilde{\mu}^2=A(\tilde{\lambda})^2$. This defines the multivalued mapping $\mathcal{L}\colon\mathbb{R}^4\to\CC^2\times\RR^2$ given by $\mathcal{L}(q,p)=(\tilde{\lambda}(q,p),\tilde{\mu}(q,p),\mathcal{EM}(q,p))$ that respects the fibers of the torus bundle. Note that the variables $\tilde\lambda$ and $\tilde\mu$ can be viewed as the coordinates of a complex reduced phase space with respect to the momentum $K$~\cite{babelon}. In particular, the orbits of $K$ are mapped to points in the image of $\mathcal{L}$ and the orbit of the Hamiltonian flow $\gamma$ to a cycle $\mathcal{L}(\gamma)$.

The different mappings are described by the commutative diagram given in Fig.~\ref{figCommDiag}.
\begin{figure}[h!]
  \centering
    \includegraphics[scale=0.55]{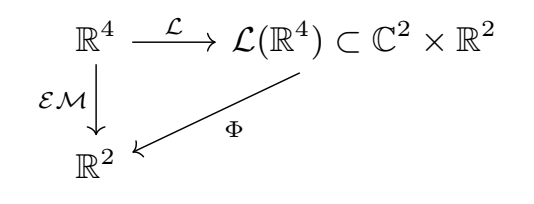}
    \caption{Commutative diagram between the mappings $\mathcal{EM}$, $\mathcal{L}$ and $\Phi$, where $\Phi$ is the projection from $\mathcal{L}(\mathbb{R}^4)\subset\mathbb{C}^2\times \mathbb{R}^2$ to $\RR^2$.}
    \label{figCommDiag}
\end{figure}
This diagram is represented schematically in Fig.~\ref{fig.diagram}. If $\mathcal{EM}$ is proper outside of the bifurcation set then, by the Arnold-Liouville theorem, the regular fibers of $\mathcal{EM}$ are real tori $\mathbb{T}^2$. In Sec.~\ref{sec.proof}, we prove that the fibers of $\Phi$ are contained in Riemann surfaces.
\begin{figure}[h!]
  \centering
    \includegraphics[scale=.25]{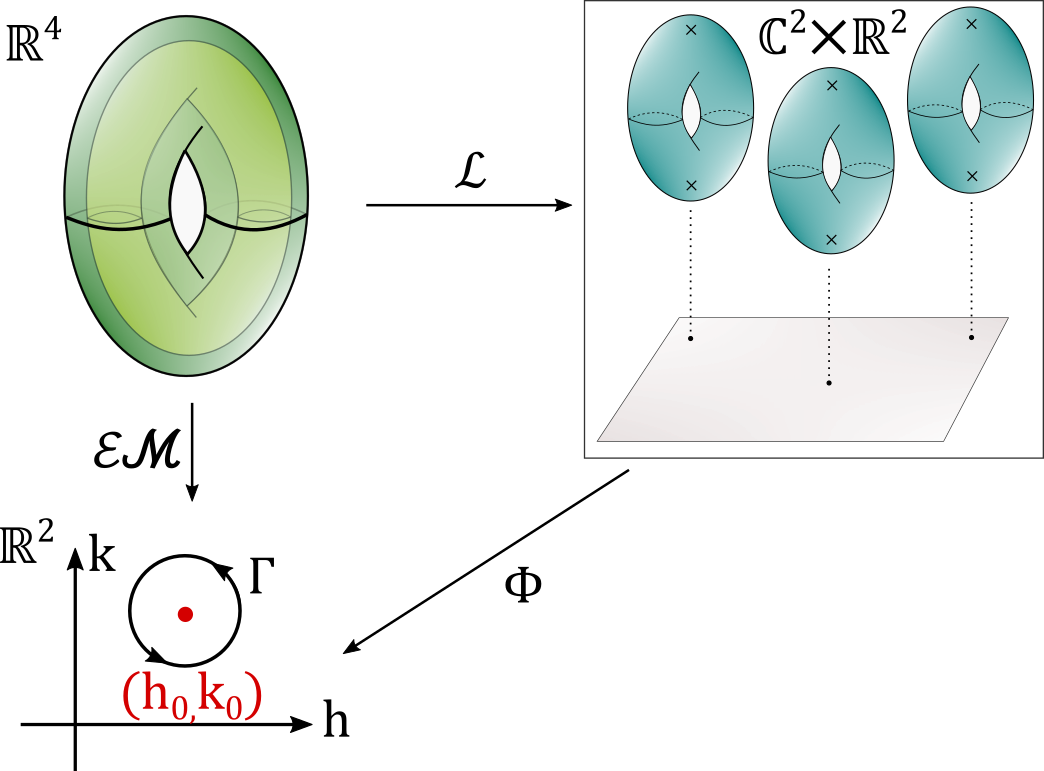}
    \caption{Schematic representation of the commutative diagram given in Fig.~\ref{figCommDiag}. The two upper panels describe respectively the fibration of the phase space by real tori $\mathbb{T}^2$ (left) and by complex tori (right). For each regular value $(h,k)$, this complex torus is a Riemann surface with two points at infinity (crosses). The lower panel depicts the image of $\mathcal{EM}$ in the space $(h,k)$. The red point indicates the position of the singular value $(h_0,k_0)$. The oriented circle represents the loop $\Gamma$ for which the Monodromy matrix is computed.}
    \label{fig.diagram}
\end{figure}

Let $(h_0,k_0)$ be an isolated critical value in the bifurcation diagram and let $\Gamma$ be a simple loop positively oriented contained in the set of regular values of $\mathcal{EM}$ around $(h_0,k_0)$, such that $(h_0,k_0)$ is the only critical value inside $\Gamma$. A first result about the structure of the Riemann surface is given by Theorem~\ref{Teo.main}.

\begin{theorem}\label{Teo.main}
Under the above hypotheses about the system, assume that $Q$ is a polynomial that has two non-real double roots for $(h,k)=(h_0,k_0)$ and fulfills the genericity condition~(G) described below.  Consider a one-form $\xi$ in the space $(\tilde{\lambda},\tilde{\mu})$ which can be written as $\xi=c_1\frac{\tilde{\lambda}d\tilde{\lambda}}{\tilde{\mu}}+c_2\frac{d\tilde{\lambda}}{\tilde{\mu}}$, $(c_1, c_2)\in\CC^2$. Let $\mathcal{I}$ be the integral $\int_{\mathcal{L}(\gamma)}\xi$, then the variation of $\mathcal{I}$ when $h$ and $k$ vary along $\Gamma$ is given by $\Delta_\Gamma \mathcal{I}=2\pi i\,\textrm{Res} (\xi,+\infty)$.\footnote{By $+\infty$ we mean the point at infinity in the upper leaf of the Riemann surface, \textit{i.e.}, $\mu>0$ for $\lambda$ large enough.} Moreover, this residue can be expressed as
$$
\textrm{Res} (\xi,+\infty)=-\frac{c_1}{\sqrt{a_4}},
$$
where $a_4$ is the leading coefficient of $Q$.
\end{theorem}
Theorem~\ref{Teo.main} is related to the HM of the system, because HM can be studied from the rotation number $\Theta$ that can be expressed as an integral of a one-form in the space $(\tilde{\lambda},\tilde{\mu})$.

More specifically, taking $\eta=\frac{\dot{\theta}}{\dot{\tilde{\lambda}}}d\tilde{\lambda}$, we get
\begin{equation*}
\Theta=\int_0^T\dot{\theta}dt=\int_{\mathcal{L}(\gamma)}\eta.
\end{equation*}
We deduce that the study of HM boils down to the study of meromorphic forms in the space $(\tilde{\lambda},\tilde{\mu})$. As a consequence of Theorem~\ref{Teo.main}, we obtain the following result.
\begin{theorem}\label{cor.main}
Consider a completely integrable Hamiltonian system in $\RR^4$, as described above, in a neighborhood of an isolated critical point $(q_0,p_0)$. Assume that  $\textrm{Res} (\eta,+\infty)=\frac{1}{i}$. Under the conditions given in Theorem~\ref{Teo.main}, the Monodromy matrix $\mathbb{M}$, for $\Gamma$, is
\begin{equation*}
\mathbb{M}=
\begin{pmatrix}
1 & 1\\
0 & 1
\end{pmatrix} ,
\end{equation*}
in an appropriate basis.
\end{theorem}

We state the main results for two degree of freedom systems but in Sec.~\ref{sec.conclu} we discuss possible generalizations of the results to higher dimensional systems.
In Sec.~\ref{subsec.Mero}, we show that the hypothesis of expressing $\xi$ as $c_1\frac{\tilde{\lambda}d\tilde{\lambda}}{\tilde{\mu}}+c_2\frac{d\tilde{\lambda}}{\tilde{\mu}}$ is natural. In Sec.~\ref{sec.exam}, we apply these results to two physical systems and we stress that the assumptions about the spectral Lax pair given by Eq.~\eqref{eq.matL} are inspired by the study of these systems. We conjecture in Sec.~\ref{secquasi} that any Hamiltonian system with a focus-focus singularity can be described, locally, by a Lax pair of this form.

\section{Technicalities}\label{sec.tech}
\subsection{Meromorphic forms on Riemann surfaces}\label{subsec.Mero}
In this section, we show that meromorphic forms on elliptic curves having at most simple poles at infinity are generated by $\{\frac{\lambda d\lambda}{\mu}$, $\frac{d\lambda}{\mu}\}$. This result is at the origin of the condition $\xi=c_1\frac{\tilde{\lambda}d\tilde{\lambda}}{\tilde{\mu}}+c_2\frac{d\tilde{\lambda}}{\tilde{\mu}}$ in the statement of Theorem~\ref{Teo.main}.

\begin{prop}\label{prop.mero}
Consider an elliptic regular curve given by the equation
$F(\lambda,\mu)=0$, of the form $F(\lambda,\mu)=\mu^2-Q(\lambda)$, with $Q$ a square-free polynomial of degree $four$. Then the space of meromorphic forms having at most simple poles at infinity is two-dimensional generated by $\{\frac{\lambda d\lambda}{\mu}$, $\frac{d\lambda}{\mu}\}$.
\end{prop}

\begin{proof}
The proof follows from the general study of the space of meromorphic functions and forms on Riemann surfaces in the spirit of the Riemann-Roch theorem. All the references can be found in the classical book~\cite{miranda1995algebraic}.

Let $X$ be a compact Riemann surface given by the compactification
\begin{equation}\label{X}
X=F^{-1}(0)\cup\{-\infty,+\infty\} ,
\end{equation}
of the complexification of $F^{-1}$, obtained by adding two points at infinity (one for each leaf). We consider $\mathcal{M}$ and $\mathcal{M}^{(1)}$ respectively the space of meromorphic functions and forms on $X$. Given $f\in\mathcal{M}$, we associate its divisor $\textrm{div}(f)=\sum_{p\in X} \textrm{ord}_{p}(f)p$. Taking $\omega\in\mathcal{M}^{(1)}$, we define similarly $\textrm{div}(\omega)$, by using local uniformizations of $\omega$ at any point $p$ of $X$.
The supports of the sum are finite. Let $D$ be a divisor. We define the vector spaces
\begin{equation*}
L(D)=\{f\in\mathcal{M}(X): \textrm{div}(f)\geq-D\} ,
\end{equation*}
and
\begin{equation*}
L^{(1)}(D)=\{\omega\in\mathcal{M}^{(1)}(X): \textrm{div}(\omega)\geq-D\} .
\end{equation*}

We are interested in one-forms having at most simple poles at infinity. Note that a form cannot have only one pole with nonzero residue as the sum of residues on a compact surface is zero. Hence, we take the divisor $D=1\cdot (+\infty)+1\cdot(-\infty)$ and consider
$L^{(1)}(D)$. Using \cite{miranda1995algebraic}~(Chapter V, Lemma 3.11), we have for any Riemann surface
\begin{equation*}
\dim L^{(1)}(D)=\dim L(D+K_d),
\end{equation*}
where $K_d$ is a canonical divisor (i.e. divisor of any meromorphic form on the Riemann surface $X$).

Now we restrict to $X$ given by Eq.~\eqref{X}. Note that the form $\omega=\frac{d\lambda}{\mu}$ is holomorphic on $X$ and does not have zero. Hence, the canonical divisor $K_d$ is the zero divisor, so for the elliptic curve we get $\dim L^{(1)}(D)=\dim L(D)$. The Riemann surface $X$ is a torus and proposition 3.14 of chapter V~\cite{miranda1995algebraic} applies and gives $\dim(L(D))=\deg(D)=2$, thus
\begin{equation*}
\dim L^{(1)}(D)=2 .
\end{equation*}
On the other hand, by the same analysis but for the 0 divisor, we obtain $\dim L^{(1)}(0)=1$ (this space is generated by the holomorphic form $\frac{d\lambda}{\mu}$).
One verifies easily that $\textrm{Res}(\frac{\lambda d\lambda}{\mu},\pm\infty)\neq 0$ (and this residue is equal to $\pm 1$ when $Q$ is unitary).
Hence, $\{\frac{\lambda d\lambda}{\mu}$, $\frac{d\lambda}{\mu}\}$ generates the space $L^{(1)}(D)$, and the differential one-forms in this space with nontrivial residues at $\pm\infty$ have a non-zero component in $\frac{\lambda d\lambda}{\mu}$.
\end{proof}

\subsection{Normal form}\label{Sec.normalform}
We formulate below the genericity condition~\eqref{eq.G} used in Theorem~\ref{Teo.main} and derive a normal form of a polynomial of degree four which verifies this condition.

\begin{lemma}\label{lem.difeo}
Let $Q(\lambda,h,k)$ be a real polynomial of degree four which for a value $(h_0,k_0)$ has two different double non-real roots. Assume that the differential of the function $F$ (defined by Eq.~\eqref{eq.F}) is invertible at $(h_0,k_0)$. Moreover, by convention we assume that\footnote{It is possible to work with the negative determinant but in this case, the change of orientation has to be taken into account.}
\begin{equation}
  \tag{G}
  \det\left[\frac{\partial F}{\partial (h,k)}\right]_{(h_0,k_0)}> 0 .
  \label{eq.G}
\end{equation}
Then there exists a local orientation preserving diffeomorphism
\begin{equation*}
(\lambda,h,k)\mapsto(\hat{\lambda}(\lambda,h,k),\hat{h}(h,k),\hat{k}(h,k))
\end{equation*}
which transforms the polynomial $Q$ to the normal form polynomial
\begin{equation*}
\hat{\lambda}^4+(2+\hat{h})\hat{\lambda}^2+\hat{k}\hat{\lambda}+1
\end{equation*}
multiplied by a unity.
\end{lemma}

\begin{proof}
First note that, since the polynomial $Q$ has real coefficients then its roots are complex conjugate. Let $\lambda_0$ and $\bar{\lambda}_0$ be the double roots for $(h,k)=(h_0,k_0)$.

After a first translation in the parameters $(h,k)$, we set $(h_0,k_0)$ to $(0,0)$. By a translation and a homothety in the variable $\lambda$, the roots $\lambda_0$, $\bar{\lambda}_0$ are given by $i$ and $-i$. We eliminate the cubic term by a translation on $\lambda$ (which depends on $h$ and $k$). In order to have the constant coefficient equal to 1, we consider a homothety on $\lambda$ and factorize, leading to the polynomial
\begin{equation*}
u(h,k)(\hat{\lambda}^4+a(h,k)\hat{\lambda}^2+b(h,k)\hat{\lambda}+1) ,
\end{equation*}
where $u$ is a unity in a neighborhood of $(h_0,k_0)$ and $a(h_0,k_0)=2$, $b(h_0,k_0)=0$.

Finally, from the transformation $(h,k)\xmapsto{F}(\hat{h},\hat{k})$ given by
\begin{equation}\label{eq.F}
F(h,k)=(a(h,k)-2,b(h,k)) ,
\end{equation}
which is a local diffeomorphism by the inverse function theorem, we arrive at the normal form polynomial
\begin{equation}\label{eq.NF}
\hat{Q}(\hat{\lambda},\hat{h},\hat{k})=\hat{\lambda}^4+(2+\hat{h})\hat{\lambda}^2+\hat{k}\hat{\lambda}+1 ,
\end{equation}
multiplied by the unity $u(F^{-1}(\hat{h},\hat{k}))$.
\end{proof}

We point out that the change of variables in $\lambda$ is global while the change of variables in $(h,k)$ is local. Thus, we can consider the transformation $(\lambda,\mu)\xmapsto{G}(\hat{\lambda},\hat{\mu})$, where $\hat{\mu}=\frac{\mu}{\sqrt{u}}$, which is a global diffeomorphism and we obtain the normal form of the spectral curve
\begin{equation*}
\hat{\mu}^2=\hat{\lambda}^4+(2+\hat{h})\hat{\lambda}^2+\hat{k}\hat{\lambda}+1 .
\end{equation*}

\subsection{Monodromy of roots of the normal form polynomial}
In this subsection, we describe the motion of the roots of the normal form polynomial obtained in Sec.~\ref{Sec.normalform}.

The Monodromy of the surfaces of the form $\mu^2=P_{h,k}(\lambda)$ is given by the movement of the roots of the polynomial $P_{h,k}(\lambda)$, when $(h,k)$ varies along a loop around a singularity. In this sense, the Monodromy of the roots codifies all the information about the Monodromy of the surfaces and the loops on them.

\begin{lemma}\label{lemma.mov}
Let $\hat{Q}(\hat{\lambda},\hat{h},\hat{k})$ be the polynomial defined by Eq.~\eqref{eq.NF}. For values of $(\hat{h},\hat{k})$ close enough to $(0,0)$, the polynomial $\hat{Q}(\hat{\lambda},\hat{h},\hat{k})$ has four roots of the form $i\pm \beta$, $-i\pm\bar{\beta}$, where $\beta\in\mathbb{C}$ depends on the value $(\hat{h},\hat{k})$. Moreover, if $(\hat{h},\hat{k})$ moves along a small loop around $(0,0)$, then the two roots close to $i$ turn around $i$ until they exchange their positions and the same statement holds for $-i$.
\end{lemma}
\begin{proof}
Consider a small circle $\mathcal{C}$ around the origin in the space $(\hat{h},\hat{k})$ with positive orientation, \textit{i.e.},
\begin{align*}
\hat{h}&= r\cos\phi \\
\hat{k}&= r\sin\phi ,
\end{align*}
with $r\ll 1$ and $\phi\in [0,2\pi]$. The polynomial
\begin{equation*}
\hat{Q}_{\hat{h},\hat{k}}(\hat{\lambda})=\hat{\lambda}^4+(2+\hat{h})\hat{\lambda}^2+\hat{k}\hat{\lambda}+1
\end{equation*}
has two double roots, $i$ and $-i$, for $(\hat{h},\hat{k})=(0,0)$. We set $\hat{\lambda}=\epsilon i+z$, where $\epsilon=\pm 1$. We substitute $\hat{\lambda}$ in the equation $\hat{Q}_{\hat{h},\hat{k}}(\hat{\lambda})=0$ and, neglecting higher order terms (using a Newton diagram), we obtain
\begin{equation}\label{eq.square}
-4z^2=\hat{h}-\epsilon i\hat{k} .
\end{equation}
From Eq.~\eqref{eq.square}, we deduce that the two roots near $i$ are $i\pm\beta$ where $\beta=\sqrt{\frac{i\hat{k}-\hat{h}}{4}}$, with an appropriate logarithmic branch. Note that the square root argument is different from 0 because $\hat{h}$ and $\hat{k}$ are both real. The roots near $-i$ are then of the form $-i\pm\bar{\beta}$ because the polynomial has real coefficients. This can also be deduced from Eq.~\eqref{eq.square}.

Now, we analyze the motion of these roots when $(\hat{h},\hat{k})$ varies along $\mathcal{C}$. Setting $z=\rho e^{i\theta}$, with $\rho\ll 1$, and substituting the parametrization of $\mathcal{C}$ we obtain
\begin{equation*}
-4\rho^2e^{2i\theta}=re^{-i\epsilon \phi} ,
\end{equation*}
which leads to
\begin{align*}
\rho&=\frac{r^{1/2}}{2} \\
\theta&=-\frac{\epsilon\phi}{2}-\frac{\pi}{2}+k\pi ,
\end{align*}
with $k\in\mathbb{Z}$. We conclude that near each double root, we obtain two simple roots which exchange their positions along the loop $\mathcal{C}$ in the space $(\hat{h},\hat{k})$.
\end{proof}

\section{Proofs of the main results}\label{sec.proof}
\begin{proof}[Proof of Theorem~\ref{Teo.main}]
The system is represented by a spectral Lax pair of the form given by Eq.~\eqref{eq.matL} and the corresponding spectral curve is
\begin{equation*}
\mu^2=A^2(\lambda)+B(\lambda)C(\lambda)
\end{equation*}
where $Q_{h,k}(\lambda)=A^2(\lambda)+B(\lambda)C(\lambda)$ is a polynomial in $\lambda$. For fixed values of the constants of motion, the spectral curve, $\mu^2=Q_{h,k}(\lambda)$, defines a Riemann surface. A link between the initial phase space and the Riemann surface can be established from the set $(\tilde{\lambda},\tilde{\mu})$ where $\tilde{\lambda}$ is a solution of the implicit equation $C(\lambda)=0$ and $\tilde{\mu}^2=A(\tilde{\lambda})^2$. \footnote{The set of coordinates $(\tilde{\lambda},\tilde{\mu})$ is related with a vector bundle that can be defined through the eigenvectors associated to the eigenvalues given by $\mu$~\cite{babelon}. Note also that we can replace $C$ by $B$ everywhere.}

First, note that $\tilde{\lambda},\tilde{\mu}\colon\mathbb{R}^4\to\mathbb{C}$ are complex valued functions that depend on the original variables of the system. The multivalued mapping $\mathcal{L}(q,p)=(\tilde{\lambda}(q,p),\tilde{\mu}(q,p),\mathcal{EM}(q,p))$ is such that $\mathcal{L}\colon\mathbb{R}^4\to\mathbb{C}^2\times\RR^2$, and we obtain the commutative diagram of Fig.~\ref{figCommDiag}.

From the definition of $\tilde{\lambda}$ and $\tilde{\mu}$, it is straightforward to deduce that
\begin{equation}\label{eq.Rie}
\tilde{\mu}^2=Q_{h,k}(\tilde{\lambda}).
\end{equation}
Equation~\eqref{eq.Rie} tells us that the fibers of $\Phi$ are contained in Riemann surfaces defined by the same equation as the spectral curve.

Now, we analyze these Riemann surfaces for a fixed value of $(h,k)\neq (h_0,k_0)$ close enough to $(h_0,k_0)$. Using Lemma~\ref{lem.difeo}, we know that the Riemann surface given by $\tilde{\mu}^2=Q_{h,k}(\tilde{\lambda})$ is diffeomorphic to the one given by the equation $\hat{\mu}^2=\hat{\lambda}^4+(2+\hat{h})\hat{\lambda}^2+\hat{k}\hat{\lambda}+1$,  and from Lemma~\ref{lemma.mov} we obtain that this polynomial has four different non-real roots, which are complex conjugate. The Riemann surface is therefore a torus with two points at infinity (see Fig.~\ref{fig.RS}).

\begin{figure}[h!]
  \centering
    \includegraphics[scale=.5]{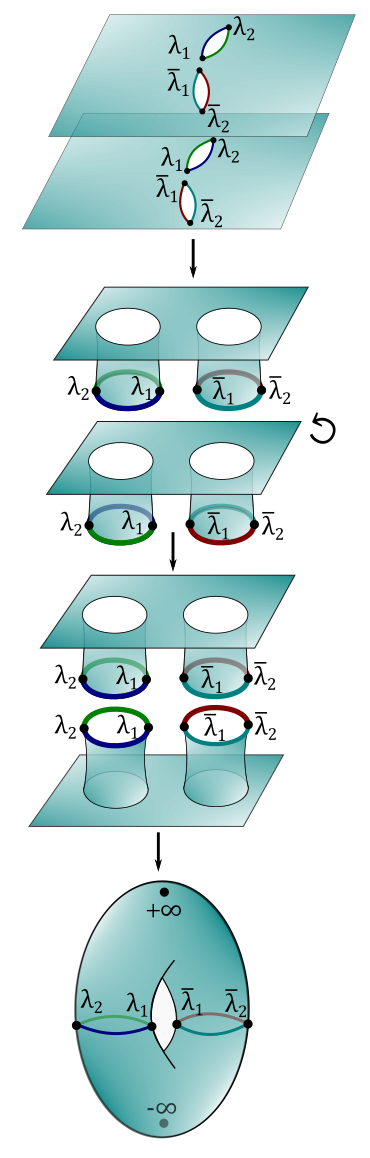}
    \caption{The upper panel represents two complex planes where cuts between the roots of the polynomial have been added. The segments with the same color in both planes are identified. Schematically, the Riemann surface can be constructed by performing the different deformations shown in this Figure.}
    \label{fig.RS}
\end{figure}
We assume that the image of the cycle $\mathcal{L}(\gamma)$ under the map $\tilde{\lambda}$ can be deformed (in the complex plane) to a simple loop that goes around two conjugate roots of $Q_{h,k}(\tilde{\lambda})$. This hypothesis can be verified directly for the two examples in Sec.~\ref{sec.exam}. In other words, $\mathcal{L}(\gamma)$ is homotopic on the Riemann surface $\tilde{\mu}^2=Q_{h,k}(\tilde{\lambda})$ to a cycle, $\tilde{\gamma}$, of the form described in Fig.~\ref{fig.cycle}.\footnote{Note that it is possible to take the opposite orientation of the cycle giving rise to a change of sign of the integral.}
\begin{figure}[h!]
  \centering
    \includegraphics[scale=.48]{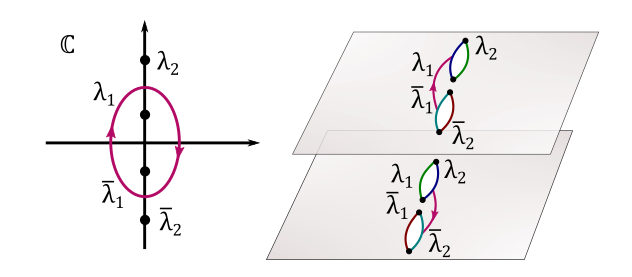}
    \caption{Configuration for the normal form with $\hat{k}=0$, $\hat{h}>0$. In the complex plane (left-hand side), the cycle is a simple loop which goes around the two conjugate roots $\lambda_1$ and $\bar{\lambda}_1$. On the right-hand side, we have the representation of this loop on the Riemann surface.}
    \label{fig.cycle}
\end{figure}

We consider the one-form $\xi$ on this Riemann surface and, since the value of the integral only depends on the homology class, we obtain
\begin{equation*}
\mathcal{I}=\int_{\mathcal{L}(\gamma)}\xi=\int_{\tilde{\gamma}}\xi ,
\end{equation*}
thus,
\begin{equation*}
\Delta_\Gamma\mathcal{I}=\int_{\tilde{\gamma}_f}\xi-\int_{\tilde{\gamma}_i}\xi=\int_{\tilde{\gamma}_f-\tilde{\gamma}_i}\xi ,
\end{equation*}
where $\tilde{\gamma}_i$ and $\tilde{\gamma}_f$ are the initial and final cycles obtained by turning around the point $(h_0,k_0)$ along $\Gamma$.

The cycle $\tilde{\gamma}_f$ can be found by following the transformation of $\tilde{\gamma}_i$ when the constants of motion vary along $\Gamma$ and this is completely determined by the movement of the roots of the polynomial $Q_{h,k}(\tilde{\lambda})$. It is then enough to analyze this movement and the result can be established from Lemma~\ref{lemma.mov}, since we know (from Lemma~\ref{lem.difeo}) that $F$ is an orientation preserving local diffeomorphism in the variables $(H,K)$ and the diffeomorphism in $(\tilde{\lambda},\tilde{\mu})$ maps the roots of one polynomial to the roots of the new one. Thus, using Lemma~\ref{lemma.mov}, we know the movement of the roots which is represented in Fig.~\ref{fig.movRoots}.
\begin{figure}[h!]
  \centering
    \includegraphics[scale=.5]{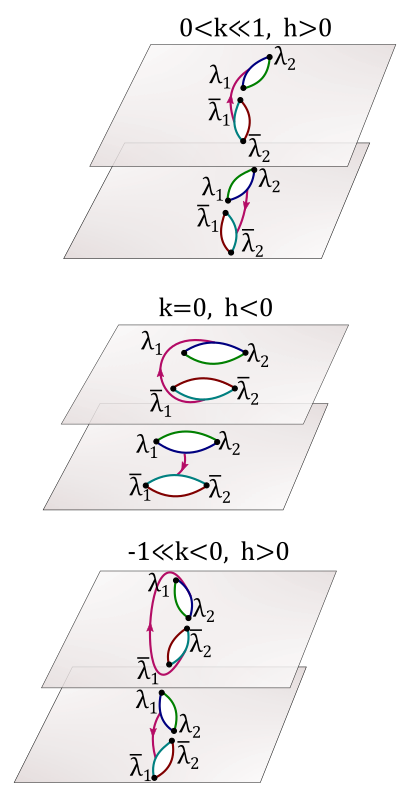}
    \caption{For $\hat{k}=0$ and $\hat{h}>0$ the roots are purely imaginary hence, they are aligned along the imaginary axis. Then, they start rotating counter clock-wise until they exchange their positions (pairwise).}
    \label{fig.movRoots}
\end{figure}
The initial and final cycles are represented on the Riemann surface in Fig.~\ref{fig.gammaIF}.
\begin{figure}[h!]
  \centering
    \includegraphics[scale=.25]{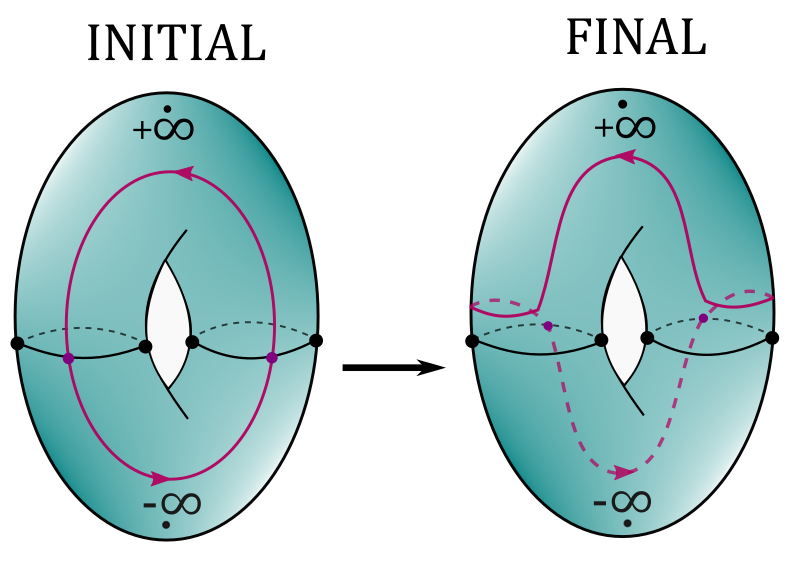}
    \caption{Analyzing the movement of the roots (Figure~\ref{fig.movRoots}), we obtain the initial and the final cycles on the Riemann surface.}
    \label{fig.gammaIF}
\end{figure}

The cycles $\tilde{\gamma}_i$ and $\tilde{\gamma}_f$ are respectively diffeomorphic to the left-hand side and right-hand side cycles in Fig.~\ref{fig.gammaIF}. Hence, the chain $\tilde{\gamma}_f-\tilde{\gamma}_f$ is homologous to the chain formed by the two cycles shown in the first torus of Fig.~\ref{fig.defoCycle}. We can deform these two cycles as shown in Fig.~\ref{fig.defoCycle} to obtain a positively oriented cycle near $+\infty$.
\begin{figure}[h!]
  \centering
    \includegraphics[scale=.49]{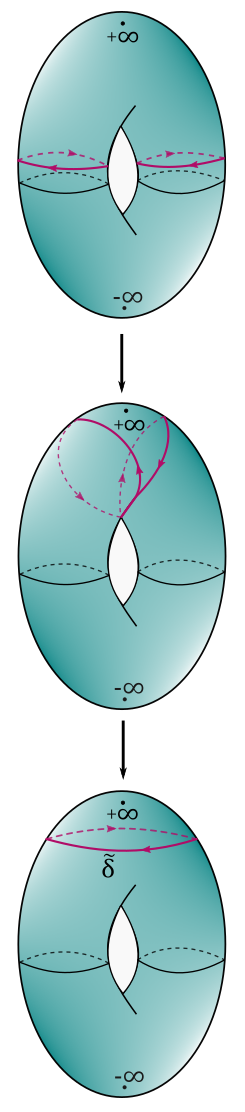}
    \caption{Deformation of the two cycles (that are homologous to $\tilde{\gamma}_f-\tilde{\gamma}_f$) into the positively oriented cycle $\tilde{\delta}$ near $+\infty$.}
    \label{fig.defoCycle}
\end{figure}

Let $\tilde{\delta}$ denote this cycle. We get
\begin{equation*}
\Delta_\Gamma\mathcal{I}=\int_{\tilde{\delta}}\xi =2\pi i\,\textrm{Res} (\xi, +\infty) ,
\end{equation*}
since, by hypothesis, $\xi=c_1\frac{\tilde{\lambda}d\tilde{\lambda}}{\tilde{\mu}}+c_2\frac{d\tilde{\lambda}}{\tilde{\mu}}$, where $(c_1, c_2)\in\CC^2$, hence, $\xi$ only has (at most) poles at infinity.

Finally, computing $\textrm{Res} (\xi, +\infty)$, we arrive at
\begin{align*}
\textrm{Res} (\xi, +\infty)&= \textrm{Res} \left(\frac{c_1\frac{1}{z}+c_2}{\sqrt{\tilde{Q}(\frac{1}{z})}}(\frac{-1}{z^2})dz,0\right)\\
&=-\frac{c_1}{\sqrt{a_4}} ,
\end{align*}
where $a_4$ is the leading coefficient of $Q(\tilde{\lambda})$.
\end{proof}

\begin{proof}[Proof of Theorem~\ref{cor.main}]
By definition of $\eta$, we have
\begin{equation*}
\int_{0}^{T}\dot{\theta}dt=\int_{\mathcal{L}(\gamma)}\frac{\dot{\theta}}{\dot{\tilde{\lambda}}}d\tilde{\lambda}=\int_{\mathcal{L}(\gamma)}\eta .
\end{equation*}
Then, using Theorem~\ref{Teo.main} we obtain that the variation of the rotation number when the constants of motion vary along $\Gamma$ is given by $\Delta_{\Gamma}\Theta=2\pi$.

A basis for the homology group of a regular torus $\mathcal{EM}^{-1}(h,k)$ can be defined as $\{\gamma_K,\gamma_H\}$, where $\gamma_K$ is the closed orbit of the flow generated by $K$ starting at a point $(q,p)\in \mathcal{EM}^{-1}(h,k)$ and $\gamma_H$ is the cycle obtained by concatenating $\gamma$ and the orbit of $K$ between this final point and the initial one (this is equivalent to follow the flow of $K$ for a time $-\Theta$).

Since we have a global $\mathbb{S}^1$-action, the first element of the basis remains the same after a turn around $\Gamma$. The second element, $\gamma_H$, transforms into itself plus the cycle $\gamma_K$, since $\Delta_{\Gamma}\Theta=2\pi$.

Thus, the Monodromy matrix in the basis described before is
\begin{equation*}
\mathbb{M}=\begin{pmatrix}
1 & 1\\
0 & 1
\end{pmatrix} .
\end{equation*}
\end{proof}

\section{Examples}\label{sec.exam}
\subsection{The Jaynes-Cummings model}
As a first example, we consider the classical Jaynes-Cummings model (JC). Its quantum counterpart has been widely studied as a basic model system in quantum optics~\cite{jaynes1963,dicke1954} describing the interaction of a two-level quantum system with a quantized mode of an optical cavity~\cite{haroche2001}. The global dynamics of the classical version has been recently explored in a series of papers showing, in particular, the possible non-trivial Monodromy of this integrable system~\cite{babelon2009,babelon2012,babelon2015,pelayo2012,dullin2019,kloc2017}. Some of these studies~\cite{babelon2012,babelon2015} use the Lax pair formalism that we propose to implement by applying the general approach presented in this paper.

The classical JC describes the interaction of a classical spin coupled to a Harmonic oscillator on the phase space $\mathbb{S}^2\times \mathbb{R}^2$. The Hamiltonian of the system can be expressed as
\begin{equation}\label{eqhamjc}
H=2\omega_0S_z+\omega \bar{b}b+g(\bar{b}S_-+bS_+)
\end{equation}
where $\omega_0$, $\omega$ and $g$ are real constants representing respectively the frequencies of the spin and of the Harmonic oscillator and the coupling strength between the two sub-systems. The coordinates describing the spin and the Harmonic oscillator are respectively denoted by $(S_x,S_y,S_z)$ and $(b,\bar{b})$ with the constraint $S_x^2+S_y^2+S_z^2=S_0^2$, where $S_0$ is a positive constant and $\bar{b}$ is the complex conjugate of $b$. To simplify the description, we also introduce the components $S_+$ and $S_-$ given by
\begin{align*}
S_+&=S_x+iS_y \\
S_-&=S_x-iS_y .
\end{align*}
The spin dynamics is obtained from the following Poisson bracket
\begin{equation*}
\{S_a,S_b\}=\varepsilon_{abc}S_c,
\end{equation*}
where $\varepsilon_{abc}$ is the completely anti-symmetric tensor with indices $a$, $b$ and $c$ belonging to the set $\{x,y,z\}$. Note that
\begin{align*}
\{S_{\pm},S_z\}&=\pm iS_{\pm} \\
\{S_+,S_-\}&=-2iS_z .
\end{align*}
For the Harmonic oscillator, we have $\{b,\bar{b}\}=-i$ which can be deduced from $b=\frac{1}{\sqrt{2}}(q+ip)$ and $\bar{b}=\frac{1}{\sqrt{2}}(q-ip)$ and the relation $\{q,p\}=1$ where $q$ and $p$ are the real position and momentum of the oscillator.

The Hamiltonian $H$ defines a Liouville integrable dynamic on a phase space of dimension four. The second constant of motion, $K$, can be written as $K=S_z+\bar{b}b$ and verifies $\{K,H\}=0$. The Hamiltonian dynamics are then governed by the following differential equations and their complex conjugate
\begin{eqnarray}\label{eqdynJC}
& &\dot{S}_+=\{S_+,H\}=2i\omega_0S_+-2ig\bar{b}S_z \nonumber \\
& &\dot{S}_z=ig\bar{b}S_--igbS_+ \\
& &\dot{b}=-i\omega b-igS_-\nonumber
\end{eqnarray}
A standard computation of the bifurcation diagram~\cite{babelon2012} shows that this system has a focus-focus singularity for some values of the parameters, and thus a non-trivial Monodromy. The corresponding isolated singular value of the bifurcation diagram has the coordinates $(h_0,k_0)=(2\omega_0S_0,S_0)$. Note that the dynamic has another fixed point in $(-2\omega_0S_0,-S_0)$. Figure~\ref{figBDJC} displays the bifurcation diagram for a specific set of parameters.
\begin{figure}[h!]
  \centering
    \includegraphics[scale=0.55]{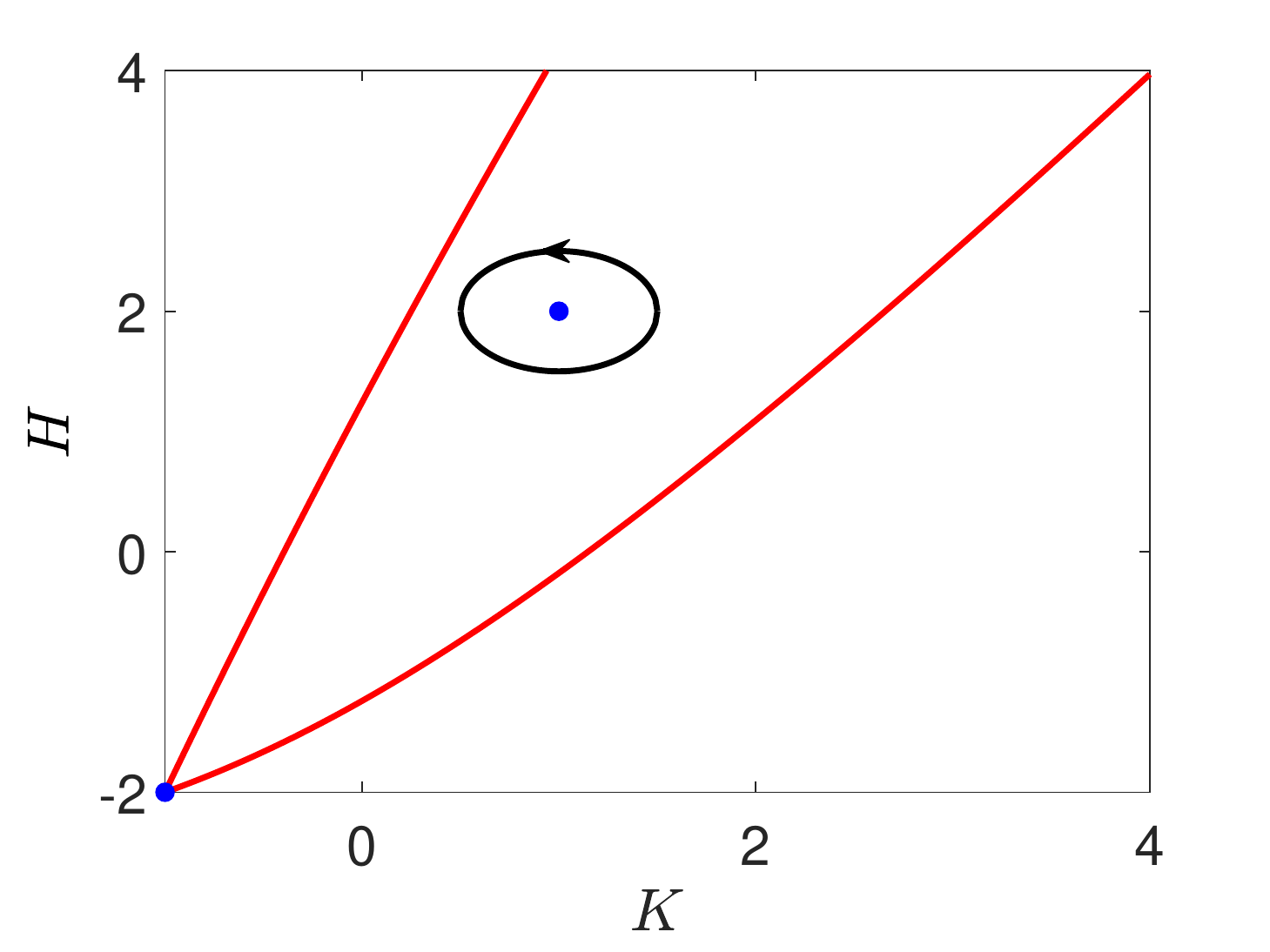}
    \caption{Plot of the JC bifurcation diagram. Numerical parameters are set to $S_0=1$, $\omega_0=1$, $\omega=2$ and $g=1$. The points for which $\mathcal{EM}$ is of ranks 0 and 1 are respectively plotted in blue and red. The blue points have the coordinates $(S_0,2\omega_0S_0)$ and $(-S_0,-2\omega_0S_0)$. A black loop around the focus-focus point is displayed in black.}
    \label{figBDJC}
\end{figure}

\noindent\textbf{The Lax Pair approach.}\\
We now describe the dynamics by using Lax pairs. We consider the Lax matrices $L$ and $M$ defined by
\begin{equation*}
L(\lambda)=
\begin{pmatrix}
\frac{(2\lambda-\omega)(\lambda-\omega_0)+g^2S_z}{g^2} & \frac{2b(\lambda-\omega_0)}{g}+S_-\\
\frac{2\bar{b}(\lambda-\omega_0)}{g}+S_+ & \frac{(\omega-2\lambda)(\lambda-\omega_0)-g^2S_z}{g^2}
\end{pmatrix}
\end{equation*}
and
\begin{equation*}
M(\lambda)=\begin{pmatrix}
-i\lambda & -igb \\
-ig\bar{b} & i\lambda
\end{pmatrix} .
\end{equation*}
The matrix $L$ has the form described by Eq.~\eqref{eq.matL}. The Lax equation reads
\begin{equation*}
\dot{L}(\lambda)=[M(\lambda),L(\lambda)],
\end{equation*}
and is equivalent to Eq.~\eqref{eqdynJC}. The spectral curve is given by
\begin{equation}\label{eqRSJC}
\mu^2=A^2+BC=Q(\lambda),
\end{equation}
where $Q(\lambda)$ is a polynomial of order four in $\lambda$ which can be expressed as
\begin{equation*}
Q(\lambda)=\frac{(2\lambda-\omega)^2}{g^4}(\lambda-\omega_0)^2+\frac{4}{g^2}K(\lambda-\omega_0)^2+\frac{2}{g^2}(H-\omega K)(\lambda-\omega_0)+S_0^2.
\end{equation*}
As mentioned in the general case, it can be verified that all the coefficients of the polynomial are constants of the motion. The movement of the roots along a loop in the bifurcation diagram can be computed numerically. To this aim, we consider a loop defined by $k=k_0+r\cos(\chi)$ and $h=h_0+r\sin(\chi)$, with $r=0.5$ and $\chi\in [0,2\pi]$. We denote by $\chi_c$ the angle defined by $\chi_c=\pi+\arctan(2\omega_0)$, which leads to the point $(h_c=h_0+r\sin(\chi_c),k_c=k_0+r\cos(\chi_c)$. The roots for this specific loop are displayed in Fig.~\ref{figrootJC}.
\begin{figure}[h]
	\centering
	\includegraphics[width = 0.55\textwidth]{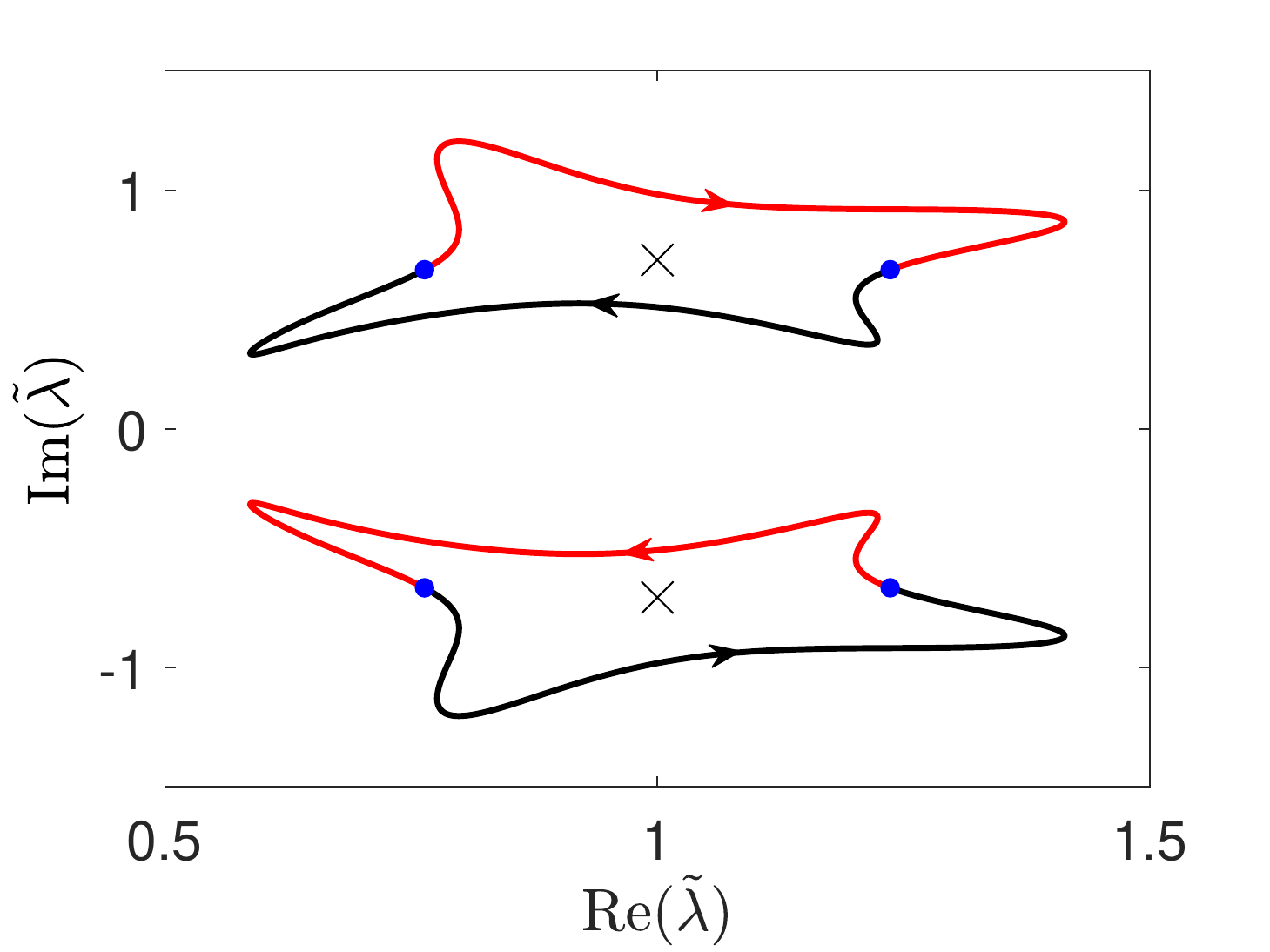}
	\caption{Plot of the movement of the roots of the polynomial $Q$ along the loop of Fig.~\ref{figBDJC}. Numerical parameters are set to $S_0=1$, $\omega_0=1$, $\omega=2$ and $g=1$. The crosses correspond to $(k_0,h_0)$ and the blue points to $(k_c,h_c)$.}
	\label{figrootJC}
\end{figure}

The Hamiltonian dynamics on this surface are obtained from the functions $(\tilde{\lambda},\tilde{\mu})$. We recall that $\tilde{\lambda}$ is the solution of the implicit equation $C(\lambda)=0$ and we consider for $\tilde{\mu}$ the component $\tilde{\mu}=A(\tilde{\lambda})$. We deduce that
\begin{eqnarray}\label{eqmulambdaJC}
& &\tilde{\lambda}=\omega_0-\frac{g}{2}\frac{S_+}{\bar{b}} ,\\
& &\tilde{\mu}=\left(\frac{\omega-2\omega_0}{2g}\right)\frac{S_+}{\bar{b}}+\frac{S_+^2}{2\bar{b}^2}+S_z .\nonumber
\end{eqnarray}
It is then straightforward to verify that
\begin{equation*}
\{\tilde{\lambda},\tilde{\mu}\}=\{-\frac{gS_+}{2\bar{b}},S_z\}=i(\tilde{\lambda}-\omega_0)
\end{equation*}
and $\{\tilde{\lambda},K\}=0$, $\{\tilde{\mu},K\}=0$, \textit{i.e.}, $\tilde{\lambda}$ and $\tilde{\mu}$ are coordinates of the reduced phase space. Figure~\ref{figflowJC} illustrates numerically the Hamiltonian trajectory on the Riemann surface. This loop can be derived either using the Hamiltonian flow in the original coordinates given by Eq.~\eqref{eqdynJC} and the relations~\eqref{eqmulambdaJC} or by using the Riemann surface~\eqref{eqRSJC}. In this latter case, we have
\begin{equation*}
\dot{\tilde{\lambda}}=\{\tilde{\lambda},H\}=\frac{g^2}{2(\tilde{\lambda}-\omega_0)}\{\tilde{\lambda},\tilde{\mu}^2\}=ig^2\tilde{\mu} ,
\end{equation*}
which leads to
\begin{equation*}
\dot{\tilde{\lambda}}=ig^2\sqrt{Q(\tilde{\lambda})} .
\end{equation*}
\begin{figure}[h]
	\centering
	\includegraphics[width = 0.55\textwidth]{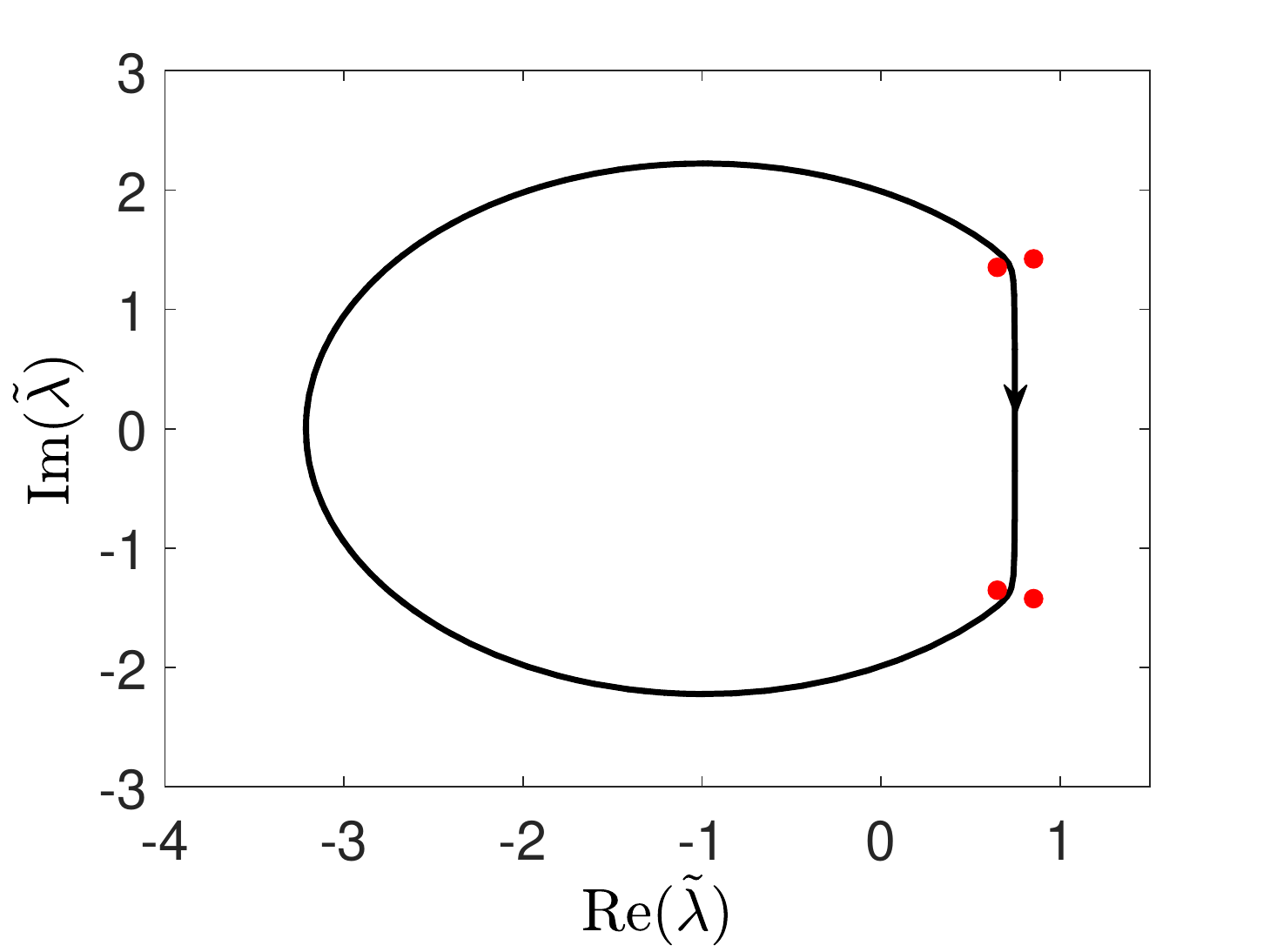}
	\caption{Plot of the Hamiltonian flow onto the Riemann surface for $k_1=k_0-0.01$ and $h_1=h_0$. The red points correspond to the position of the branch points of the Riemann surface.}
	\label{figflowJC}
\end{figure}

\noindent\textbf{The normal form.}\\
The next step of our general approach consists in analyzing the polynomial $Q_{h,k}(\tilde{\lambda})$ around the point $(h_0,k_0)$ in the bifurcation diagram. Since $Q\in\mathbb{R}[\tilde{\lambda}]$, we have $\overline{Q(\tilde{\lambda})}=Q(\bar{\tilde{\lambda}})$ which implies that the roots of $Q_{k,h}$ are pairwise complex conjugate. We first compute $Q_{h_0,k_0}(\tilde{\lambda})$ and its roots:
\begin{align*}
Q_{h_0,k_0}(\tilde{\lambda})&=\left(\frac{(2\tilde{\lambda}-\omega)(\tilde{\lambda}-\omega_0)}{g^2}+S_0\right)^2\\
&=\frac{4}{g^4}\left(\tilde{\lambda}^2-\frac{2\omega_0+\omega}{2}\tilde{\lambda}+\frac{\omega_0\omega+g^2S_0}{2}\right)^2
\end{align*}
Since this polynomial has conjugate roots, we also have
\begin{align*}
Q_{h_0,k_0}(\tilde{\lambda})&=\frac{4}{g^4}((\tilde{\lambda}-\lambda_0)(\tilde{\lambda}-\bar{\lambda}_0))^2\\
&=\frac{4}{g^4}(\tilde{\lambda}^2-2\textrm{Re\,}(\lambda_0)\tilde{\lambda}+||\lambda_0||^2)^2,
\end{align*}
and we arrive at
\begin{align*}
\textrm{Re\,}(\lambda_0)&=\frac{2\omega_0+\omega}{4}\\
\textrm{Im\,}(\lambda_0)&=\left(\frac{g^2S_0}{2}-\left(\frac{\omega-2\omega_0}{4}\right)^2\right)^{\frac{1}{2}}.
\end{align*}
Then, we apply the general procedure described in Sec.~\ref{sec.tech} to transform this polynomial to the one of the normal form. We have:
\begin{eqnarray*}
& & Q_{h,k}(\tilde{\lambda})=\frac{4}{g^4}\tilde{\lambda}^4-\frac{4}{g^4}(\omega+2\omega_0)\tilde{\lambda}^3\\
& & +\frac{4}{g^4}((\omega_0+\frac{\omega}{2})^2+\omega\omega_0+g^2k)\tilde{\lambda}^2\\
& & -\frac{4}{g^4}(\omega\omega_0(\omega_0+\frac{\omega}{2})+\frac{g^2}{2}(k(\omega+4\omega_0)-h))\tilde{\lambda}\\
& & +\frac{4}{g^4}(\frac{\omega^2\omega_0^2}{4}+\frac{g^4S_0^2}{4}+\frac{\omega_0g^2}{2}(k(\omega+2\omega_0)-h))
\end{eqnarray*}
From a direct calculus, we obtain that the derivative at $(h_0,k_0)$ of the transformation $F$ defined in Lemma~\ref{lem.difeo} has the following form
\begin{equation*}
D=
\begin{pmatrix}
\frac{g^2(\omega_0- \textrm{Re} (\lambda_0))}{2 \textrm{Im} (\lambda_0)^4} & \frac{g^2(s)}{2 \textrm{Im} (\lambda_0)^4}\\
\frac{g^2}{2 \textrm{Im} (\lambda_0)^3} & \frac{g^2(2-(\omega+4\omega_0))}{2 \textrm{Im} (\lambda_0)^3}
\end{pmatrix}
\end{equation*}
where
\begin{equation*}
s=2 \textrm{Im} (\lambda_0)^2-2 \textrm{Re} (\lambda_0)^2 + \textrm{Re} (\lambda_0) (\omega+4\omega_0) - \omega_0(\omega+2\omega_0).
\end{equation*}
The determinant of the matrix $D$ is
\begin{eqnarray*}
& &  \det(D)=\frac{g^4}{4 \textrm{Im} (\lambda_0)^7}( -\textrm{Im} (\lambda_0)^2+2 \textrm{Re} (\lambda_0)^2 \\
& & - \textrm{Re} (\lambda_0) (\omega+4\omega_0) + \omega_0(\omega+2\omega_0)\\
& &+(\omega_0- \textrm{Re} (\lambda_0))(2-(\omega+4\omega_0))).
\end{eqnarray*}
Substituting the values of $\textrm{Re} (\lambda_0)$ and $\textrm{Im} (\lambda_0)$, we obtain
\begin{equation*}
\det D=\frac{g^4(-g^2 S_0-\frac{\omega_0^2}{2}-\frac{\omega}{2}+\frac{\omega^2}{4}+\omega_0)}{4 \textrm{Im} (\lambda_0)^7}.
\end{equation*}
If $\det D\neq 0$, then $Q$ can be transformed to the normal form polynomial.\\
\noindent\textbf{Hamiltonian Monodromy.}\\
We have now all the tools in hand to compute the Monodromy matrix. The angle $\theta$ conjugate to $K$ verifies $\{\theta,K\}=1$ and $\{\theta,\tilde{\lambda}\}=\{\theta,\tilde{\mu}\}=0$. This result can be derived from the original coordinates of the phase space as shown in Appendix~\ref{appA}. Using Eq.~\eqref{eqRSJC}, we have
\begin{equation*}
\dot{\theta}=\{\theta,H\}=\{\theta,-2K(\tilde{\lambda}-\omega_0)+\omega K\}
\end{equation*}
which gives
\begin{equation*}
\dot{\theta}=\omega+2\omega_0-2\tilde{\lambda}
\end{equation*}
and we obtain the following expression for the one-form $d\theta$ as a function of $\tilde{\lambda}$ and $\tilde{\mu}$
\begin{equation}\label{dthetaJC}
d\theta=\frac{((\omega+2\omega_0)-2\tilde{\lambda})d\tilde{\lambda}}{ig^2\tilde{\mu}} .
\end{equation}
From Eq.~\eqref{dthetaJC}, it is clear that $d\theta$ is of the form described in Theorem~\ref{Teo.main}. We deduce that the residue of this form at infinity is equal to $\frac{2}{ig^2}$ over the square root of the leading coefficient of the polynomial $Q$, \textit{i.e.},
\begin{equation*}
\frac{\frac{2}{ig^2}}{\sqrt{\frac{4}{g^4}}}=\frac{2g^2}{2g^2i}=\frac{1}{i}.
\end{equation*}
Finally, using Theorem~\ref{cor.main} we obtain that the Monodromy matrix is
\begin{equation*}
\mathbb{M}=
\begin{pmatrix}
1 & 1\\
0 & 1
\end{pmatrix} .
\end{equation*}
\subsection{The spherical pendulum}
As a second example, we consider the historical system for which a non-trivial Monodromy was for the first time highlighted~\cite{duistermaat}, namely the spherical pendulum. Since then, this system has been extensively studied both from a classical and quantum points of view~\cite{efstathioubook,cushmannbook,cusman1988,beukers,audin2002,guillemin1989}. The corresponding Lax pair formalism has been described in~\cite{audin2002,gavrilov2002}. The spherical pendulum consists in a mass moving without friction on a sphere. The dynamics are governed on the phase space $TS^2$ by the following Hamiltonian expressed in dimensionless coordinates~\cite{cushmannbook}:
\begin{equation*}
H=\frac{1}{2}(p_x^2+p_y^2+p_z^2)+z
\end{equation*}
with the constraints $x^2+y^2+z^2=1$ and $xp_x+yp_y+zp_z=0$. The
dynamical equations can be expressed as
\begin{align*}
\dot{\vec{q}}&=\vec{p} \\
\dot{\vec{p}}&=\vec{e}_z-(p_x^2+p_y^2+p_z^2-z)\vec{q}
\end{align*}
where $\vec{e}_z=(0,0,-1)$ is a unit vector along the $z$- direction. We introduce the angular momentum $\vec{K}=\vec{q}\times \vec{p}$ with $\dot{\vec{K}}=\vec{q}\times\vec{e}_z$. We deduce that the system is Liouville-integrable since it has a second constant of motion, $K=K_z$ such that $\{H,K\}=0$. It can be shown that the spherical pendulum has a non-trivial Monodromy due to a focus-focus singularity corresponding to the point $(h_0,k_0)=(1,0)$ of the bifurcation diagram~\cite{efstathioubook,cushmannbook}.

\noindent\textbf{The Lax Pair approach.}\\
We introduce the coordinates $(L_x,L_y,L_z)$ and $(M_x,M_y,M_z)$ defined as
\begin{eqnarray*}
& & L_x=x-\lambda K_x \\
& & L_y=y-\lambda K_y \\
& & L_z=z-\lambda K_z+\lambda^2
\end{eqnarray*}
and
\begin{eqnarray*}
& & M_x=K_x \\
& & M_y=K_y \\
& & M_z=K_z-\lambda
\end{eqnarray*}
The Lax matrices $L$ and $M$ which satisfy $\dot{L}=[M,L]$ can then be expressed as:
\begin{eqnarray*}
& & L=L_x\sigma_x+L_y\sigma_y+L_z\sigma_z \\
& & M=\frac{-i}{2}(M_x\sigma_x+M_y\sigma_y+M_z\sigma_z),
\end{eqnarray*}
where $\sigma_{x,y,z}$ are the Pauli matrices. The spectral curve is given by the eigenvalues of $L$
\begin{equation*}
\mu^2=L_x^2+L_y^2+L_z^2.
\end{equation*}
A straightforward computation leads to
\begin{equation}\label{eqSPRS}
\mu^2=\lambda^4-2K\lambda^3+2H\lambda^2+1.
\end{equation}
We have
\begin{align*}
A(\lambda)&=L_z \\
C(\lambda)&=L_x-iL_y
\end{align*}
and we deduce that
\begin{equation*}
\tilde{\lambda}=\frac{x+iy}{K_x+iK_y}
\end{equation*}
and
\begin{equation*}
\tilde{\mu}=A(\tilde{\lambda})=z-\tilde{\lambda} K_z+\tilde{\lambda}^2
\end{equation*}
As expected, it can be verified that
\begin{equation*}
\{\tilde{\lambda},K\}=\{\tilde{\mu},K\}=0
\end{equation*}
Using $\{K_x,z\}=-y$ and $\{K_y,z\}=x$, we get
\begin{equation*}
\{\tilde{\lambda},\tilde{\mu}\}=\{\tilde{\lambda},z\}=-i\tilde{\lambda}^2.
\end{equation*}
Finally, from Eq.~\eqref{eqSPRS}, we arrive at
\begin{equation*}
\dot{\tilde{\lambda}}=\{\tilde{\lambda},H\}=\frac{\tilde{\mu}}{\tilde{\lambda}^2}\{\tilde{\lambda},\tilde{\mu}\}=-i\tilde{\mu}.
\end{equation*}

\noindent\textbf{The normal form.}\\
As in the JC model, we apply the procedure described in Sec.~\ref{sec.tech} to transform the polynomial
\begin{equation*}
Q_{h,k}(\tilde{\lambda})=\tilde{\lambda}^4-2k\tilde{\lambda}^3+2h\tilde{\lambda}^2+1
\end{equation*}
to the normal form polynomial. The polynomial $Q_{h_0,k_0}(\tilde{\lambda})$ given by
\begin{equation*}
Q_{h_0,k_0}(\tilde{\lambda})=\tilde{\lambda}^4+2\tilde{\lambda}^2+1 ,
\end{equation*}
has the roots $\pm i$. Following Lemma~\ref{lem.difeo}, we obtain that the derivative at $(h_0,k_0)$ of the transformation $F$ is
\begin{equation*}
D=\begin{pmatrix}
0 & 2\\
2 & 0
\end{pmatrix} .
\end{equation*}
Since $\det D\neq 0$, the polynomial $Q$ can be transformed to the normal form.\\
\noindent\textbf{Hamiltonian Monodromy.}\\
The last step to compute the Monodromy matrix consists in expressing the one-form $d\theta$ in the coordinates $(\tilde{\lambda},\tilde{\mu})$. The angle $\theta$ conjugate to $K$ verifies $\{\theta,K\}=1$ and $\{\theta,\tilde{\lambda}\}=\{\theta,\tilde{\mu}\}=0$. Note that the expression of $\theta$ as a function of the original coordinates of the phase space is derived in Appendix~\ref{appB}. Starting from Eq.~\eqref{eqSPRS}, we deduce that
\begin{equation*}
\dot{\theta}=\{\theta,H\}=\tilde{\lambda}.
\end{equation*}
Finally, we arrive at
\begin{equation}\label{eq1formSP}
d\theta=\dot{\theta}\frac{d\tilde{\lambda}}{\dot{\tilde{\lambda}}}= i\frac{\tilde{\lambda}d\tilde{\lambda}}{\tilde{\mu}}.
\end{equation}
The one-form given in Eq.~\eqref{eq1formSP} corresponds to the expression described in Theorem~\ref{Teo.main} and its residue at infinity is equal to $-i=\frac{1}{i}$. Using Theorem~\ref{cor.main} we conclude that the Hamiltonian Monodromy matrix is
\begin{equation*}
\mathbb{M}=
\begin{pmatrix}
1 & 1\\
0 & 1
\end{pmatrix} .
\end{equation*}
\section{Quasi Lax pair of a system with a focus-focus singularity}\label{secquasi}
We have shown in two examples that the general method described in this paper can be applied to classical systems whose Lax pairs are known. Another issue is to prove with this complex approach the non-trivial Monodromy of any system with a focus-focus singularity. We come up against the difficulties of deriving Lax pairs for which no general algorithm exists. We propose in this section a partial answer to this open question by constructing a quasi-Lax pair of the dynamical system, in the sense that the dynamical equations are only reproduced up to some terms (which are negligible in a neighborhood of the focus-focus singularity) by the Lax equation.

For that purpose, we consider a two-degree of freedom system with the coordinates $(a,\bar{a})$ and $(b,\bar{b})$ such that $\{a,\bar{a}\}=\{b,\bar{b}\}=-i$. The Hamiltonian defined as
\begin{equation*}
H=ab+\bar{a}\bar{b}
\end{equation*}
is Liouville-integrable with the constant of the motion $K$ given by $K=\bar{b}b-a\bar{a}$. This system corresponds to a 1:-1 resonant system which is the non-compact local normal form for dynamics with a single focus-focus point in $a=b=0$. The nontrivial Monodromy of the Hamiltonian dynamics is e.g. shown in~\cite{efstathiou2017} with standard techniques. The Hamiltonian dynamics are governed in a neighborhood $V$ of the focus-focus point by the following differential equations
\begin{eqnarray}\label{eqab}
\dot{a}=-i\bar{b};~\dot{\bar{a}}=ib \nonumber \\
\dot{b}=-i\bar{a};~\dot{\bar{b}}=ia
\end{eqnarray}
The quasi Lax pair is defined by the two matrices $L$ and $M$ given by
\begin{equation*}
L=(\lambda^2+1-a\bar{a}/2)\sigma_z+(b\lambda+\bar{a})\sigma_++(\bar{b}\lambda+a)\sigma_-
\end{equation*}
and
\begin{equation*}
M=\frac{1}{2}(-i\lambda \sigma_z-ib\sigma_+-i\bar{b}\sigma_-)
\end{equation*}
where we use the Pauli matrices $\sigma_{x,y,z}$ and $\sigma_\pm=\frac{1}{2}(\sigma_x\pm i\sigma_y)$. We find the equations of motion from $\dot{L}=[M,L]$ in which
terms of order larger than (or equal to) 3 are neglected. Note the similarity between this derivation and the one for the JC model.

In this case, we define (following the usual formalism) the spectral curve by $\det(L(\lambda)-\mu I)=0$ and the new set of coordinates $(\tilde{\lambda},\tilde{\mu})$ where $\tilde{\mu}=A(\tilde{\lambda})$ and $\tilde{\lambda}$ is the solution of the implicit equation $C(\lambda)=0$.

We stress that the spectral curve
\begin{equation}\label{eqscnew}
\tilde{\mu}^2=\tilde{\lambda}^4+(2+k)\tilde{\lambda}^2-h\tilde{\lambda}+1-\frac{a^2\bar{a}^2}{4}
\end{equation}
depends on time (not all its coefficients are constants of motion) but if we neglect terms of order larger than three we directly obtain the normal form of the spectral curve derived in Sec.~\ref{Sec.normalform}.

At this level appears a major difficulty in the method proposed in this Section, the last term of the right-hand side of Eq.~\eqref{eqscnew}, $-\frac{a^2\bar{a}^2}{4}$, causes the Riemann surface to be different for each point of the Hamiltonian flow. More specifically, the image of each point of the Hamiltonian flow under the mapping $(\tilde{\lambda},\tilde{\mu})$ belongs to a different Riemann surface which is a small perturbation of the one given by the normal form of the spectral curve. This fact makes the use of the spectral curve more difficult.

However, the change of variables
\begin{eqnarray*}
& & \tilde{\lambda}=-\frac{a}{\bar{b}} ,\\
& & \tilde{\mu}=\tilde{\lambda}^2+1-\frac{a\bar{a}}{2} ,
\end{eqnarray*}
achieved from the quasi-Lax pair, can be applied to derive a new expression of the Liouville two-form. This latter can either be expressed as
\begin{equation*}
\Omega=ida\wedge d\bar{a}+idb\wedge d\bar{b}
\end{equation*}
or
\begin{equation}\label{eqtwoformq}
\Omega=-id\ln(\bar{b})\wedge dK-\frac{2i}{\tilde\lambda}d\tilde\lambda\wedge d\tilde\mu
\end{equation}
and, using Eq.~\eqref{eqtwoformq}\footnote{Note the similarity of this expression with the ones derived in the appendices for the JC model and the Spherical Pendulum.}, we obtain the one-form $d\theta=-id\ln(\bar{b})$ which leads to $\dot\theta=-\tilde{\lambda}$. We also have $\{\tilde\lambda,\tilde\mu\}=\frac{i}{2}\tilde\lambda$ and, using Eq.~\eqref{eqab}, we arrive at
\begin{equation}
\dot{\tilde{\lambda}}=\{\tilde\lambda,H\}=i(1+\tilde{\lambda}^2).
\end{equation}

In order to analyze the monodromy of the system in a small neighborhood $V$ of the focus-focus point (we can consider, e.g., a ball of radius $R$ small enough such that $\|a\|^2+\|b\|^2\leq R^2$), we have to work with relative cycles (see, e.g.,~\cite{efstathiou2017} for details). Thus, the rotation number $\Theta$ can be expressed as
\begin{equation}\label{eqquasitheta}
\Theta=\int_{\gamma_{rel}} d\theta=\int_{\tilde{\gamma}_{\textrm{rel}}}\frac{-\tilde\lambda}{i(1+\tilde{\lambda}^2)}d\tilde\lambda,
\end{equation}
where the relative cycle $\gamma_{\textrm{rel}}$ is the part of the Hamiltonian flow contained in $V$ and its image under the mapping $\tilde{\lambda}$ is $\tilde{\gamma}_{\textrm{rel}}$. The variation of the rotation number is given by
\begin{equation*}
\Delta\Theta=\int_{\Delta\tilde{\gamma}_{rel}}\frac{-\tilde\lambda}{i(1+\tilde{\lambda}^2)}d\tilde\lambda.
\end{equation*}
It is then enough to study the behavior of $\tilde{\lambda}$ over the Hamiltonian flow in $V$.

Integrating Eq.~\eqref{eqab}, we obtain that the Hamiltonian flow in $V$ is given by
\begin{eqnarray*}
a(t)=a_0\cosh(t)-i\bar{b}_0\sinh(t) \\
\bar{b}(t)=\bar{b}_0\cosh(t)+ia_0\sinh(t)
\end{eqnarray*}
where $a_0$ and $\bar{b}_0$ are the initial values at $t=0$ of $a$ and $\bar{b}$. The next step consists in describing the Hamiltonian flow in the new variable $\tilde{\lambda}$. This behavior is given by the following Lemma.
\begin{lemma}\label{lemma.prope}
The mapping $\tilde{\lambda}$, restricted to $V$, has the following properties
\begin{enumerate}
\item On the singular fiber $\mathcal{EM}^{-1}(0,0)\setminus \{(0,0,0,0)\}$, the function $\tilde{\lambda}$ takes two values $i$ and $-i$ (one on each connected component).
\item In a regular fiber $\mathcal{EM}^{-1}(h,k)$, $\tilde{\lambda}$ does not take the value $i$ nor $-i$ and it is injective on each Hamiltonian trajectory.
\item $\tilde{\lambda}$ is well-defined and continuous on each regular fiber, except  for those fibers such that $h=0$ and $k<0$ where $\tilde{\lambda}$ takes the value $\infty$.\footnote{In other words, if we consider $\tilde{\lambda}\colon V\to\hat{\CC}$, then it is well defined and continuous on each regular fiber, where $\hat{\CC}$ is the Riemann sphere.}
\item In a regular fiber $\mathcal{EM}^{-1}(h,k)$, the image of the Hamiltonian flow is a curve that starts and ends respectively near $-i$ and $i$. When $(h,k)$ goes to $(0,0)$, the initial and final points converge to $-i$ and $i$.
\end{enumerate}
\end{lemma}
\begin{proof}[Proof of Lemma~\ref{lemma.prope}, 1.]
For the singular fiber for which $h=k=0$, the initial condition $(0,0)\neq(a_0,b_0)\in V$ of the Hamiltonian flow fulfills
\begin{align}
||a_0||=||b_0|| ,\label{eq.Norm}\\
\textrm{Re}(a_0b_0)=0 .\label{eq.RePart}
\end{align}
Note that Eq.~\eqref{eq.Norm} implies that $a_0\neq0$ and $b_0\neq0$. Using Eq.~\eqref{eq.RePart}, we get $a_0b_0=ir$ with $r\in\mathbb{R}$, which, from Eq.~\eqref{eq.Norm} leads to $\|r\|=\|b_0\|^2$. Since $\bar{b}_0=\frac{\|b_0\|^2}{b_0}$, we obtain
\begin{equation*}
a_0=\pm i\bar{b}_0 ,
\end{equation*}
for $\|r\|=\pm r$. We use this relation to compute the image of a solution curve, $\varphi(t)$, contained in the singular fiber
\begin{align*}
\tilde{\lambda}(\varphi(t))&=\frac{i\bar{b}_0\sinh(t)-a_0\cosh(t)}{\bar{b}_0\cosh(t)+ia_0\sinh(t)}\\
&=\frac{i\bar{b}_0\sinh(t)\mp i\bar{b}_0\cosh(t)}{\bar{b}_0\cosh(t)\pm i\bar{b}_0i\sinh(t)}\\
&=\mp i.
\end{align*}
We conclude that $\tilde{\lambda}$ is constant for the trajectories contained in the singular fiber and that it only takes two values, $i$ (when $a_0=- i\bar{b}_0$) and $-i$ (when $a_0=i\bar{b}_0$).
\end{proof}
The proofs of the other points of Lemmas~\ref{lemma.prope} follow the same lines as the first point. The mapping $\tilde{\lambda}$ is schematically represented in Fig.~\ref{figpinched}.
\begin{figure}[h]
	\centering
	\includegraphics[width = 0.55\textwidth]{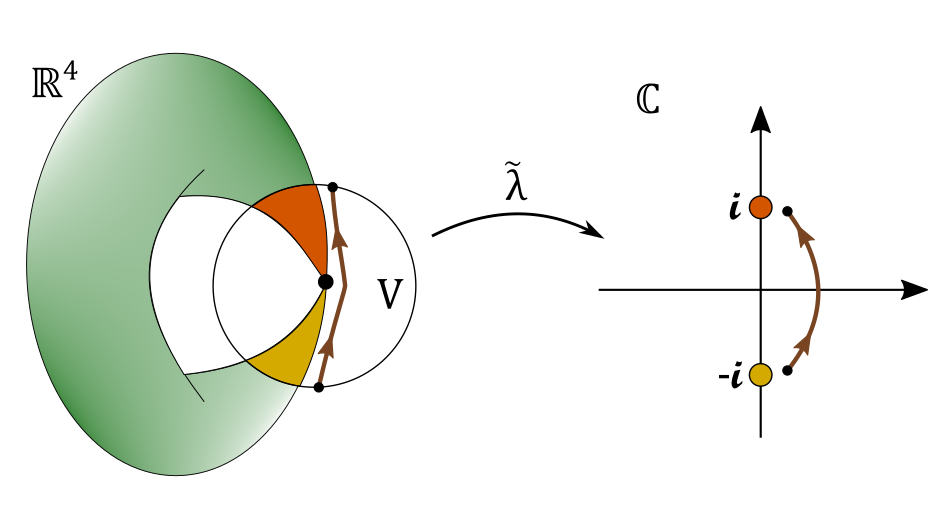}
	\caption{Schematic representation of the mapping $\tilde{\lambda}$ in a neighborhood $V$ of the focus-focus point.}
	\label{figpinched}
\end{figure}
The image of a regular fiber $\mathcal{EM}^{-1}(h,k)$ under $\tilde{\lambda}$ can be described as follows. For $h>0$ (resp. $h<0$) and $k=0$, the image is contained in the left (resp. right) part of the unitary circle. For $k>0$ (resp. $k<0$), the norm of the points in the image of the fiber is smaller (resp. larger) than 1. For $h=0$ and $k>0$, the image is contained in the imaginary axis between $i$ and $-i$, while it belongs to the imaginary axis outside of the segment between $i$ and $-i$ for $h=0$ and $k<0$. We point out that there is no intersection between the different trajectories. Indeed, they are given as solutions of the differential equation $\dot{\lambda}=1+\lambda^2$. Hence, by the uniqueness theorem for differential equations, they can not intersect unless they coincide, which is not the case.
\begin{figure}[h]
	\centering
	\includegraphics[width = 0.55\textwidth]{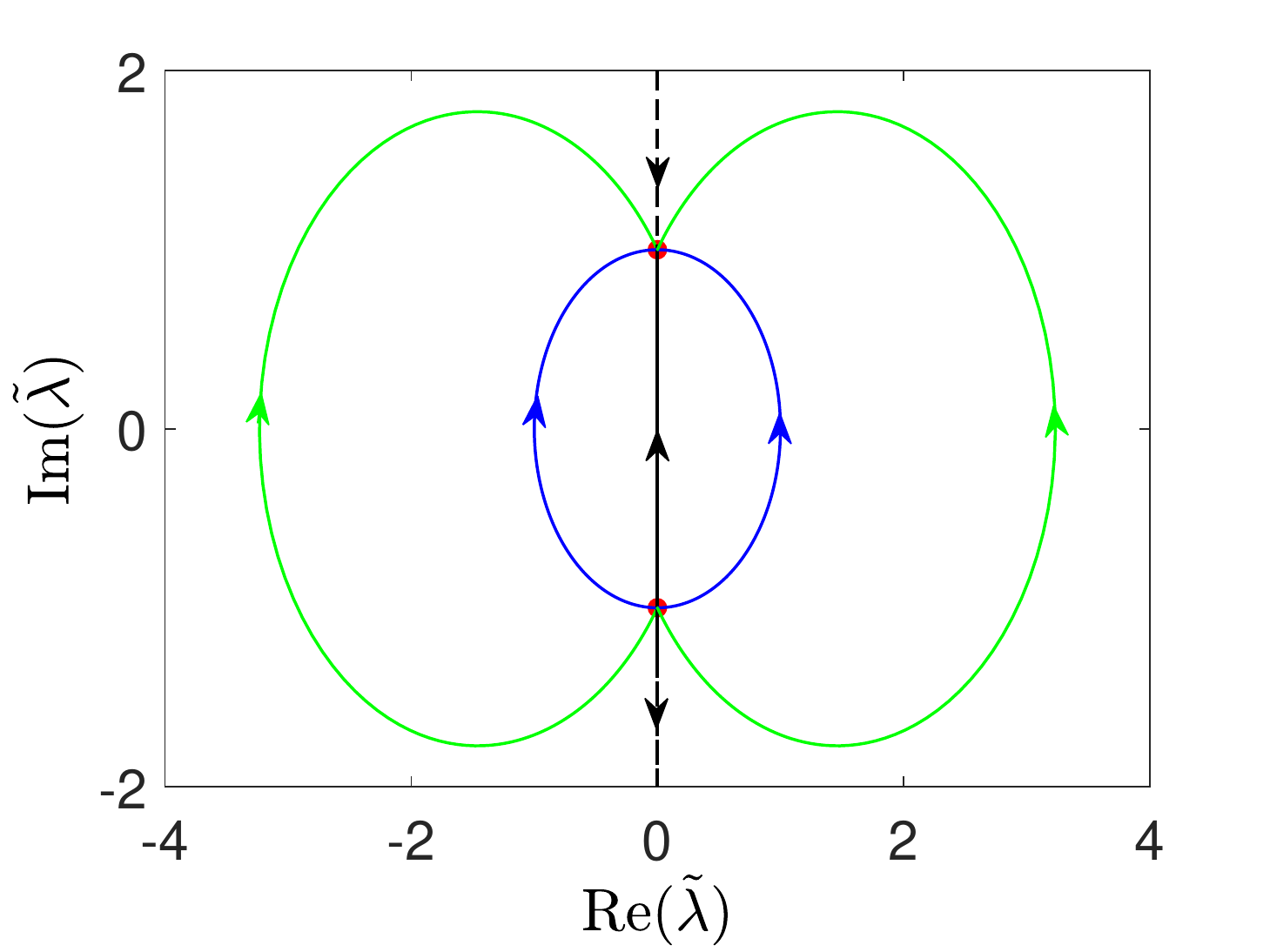}
	\caption{Plot of the trajectories in the complex plane of $\tilde{\lambda}$. The red points correspond to $+i$ and $-i$. The blue and green curves represent respectively the Hamiltonian flows for $\phi=0$, $\pi$ and $\frac{3\pi}{2}-\epsilon$, $-\frac{\pi}{2}+\epsilon$. The solid and dashed black lines depict respectively the trajectories for $\phi=\frac{\pi}{2}$ and $-\frac{\pi}{2}$. The trajectories converge to the points $\pm i$ when $t\to \pm \infty$. Numerical values are set to $\rho=0.1$, $a_0=0.5$ and $\epsilon=0.6$.}
	\label{figquasi}
\end{figure}

We consider now a small circle, $\Gamma$, of radius $\rho$ in the space of parameters $(h,k)$ starting in $h=0$, $k<0$ and positively oriented. The previous analysis gives the behavior of $\tilde{\gamma}_{\textrm{rel}}$ when $(h,k)$ vary along $\Gamma$. These results are illustrated numerically in Fig.~\ref{figquasi}. Thus, we have all the tools in hand to compute the variation of the rotation number along $\Gamma$
\begin{equation*}
\Delta_{\Gamma}\Theta=\int_{\Delta_{\Gamma}\tilde{\gamma}_{\textrm{rel}}}\frac{-\tilde\lambda}{i(1+\tilde{\lambda}^2)}d\tilde\lambda.
\end{equation*}
Note that the one-form has a pole at infinity and for the value $h=0$, $k<0$, $\tilde{\gamma}_{\textrm{rel}}$ goes through this point. For this reason, we parameterize $\Gamma$ as $h+ik=\rho e^{i\phi}$ with $\phi\in[-\frac{\pi}{2},\frac{3\pi}{2}]$ and we consider the limit
\begin{align*}
\Delta_{\Gamma}\Theta&=\lim_{\epsilon\to 0^+}\left(\Theta(\rho e^{i(\frac{3\pi}{2}-\epsilon)})-\Theta(\rho e^{i(\frac{-\pi}{2}+\epsilon)})\right)\\
&=\lim_{\epsilon\to 0^+}\left(\int\limits_{\Delta_\epsilon\Gamma_{rel}}\frac{-\tilde\lambda}{i(1+\tilde{\lambda}^2)}d\tilde\lambda\right)
\end{align*}
where $\Delta_\epsilon\Gamma_{\textrm{rel}}=\tilde{\gamma}_{\textrm{rel}}(\rho e^{i(\frac{3\pi}{2}-\epsilon)})-\tilde{\gamma}_{rel}(\rho e^{i(\frac{-\pi}{2}+\epsilon)})$.

\begin{figure}[h]
	\centering
	\includegraphics[width = 0.55\textwidth]{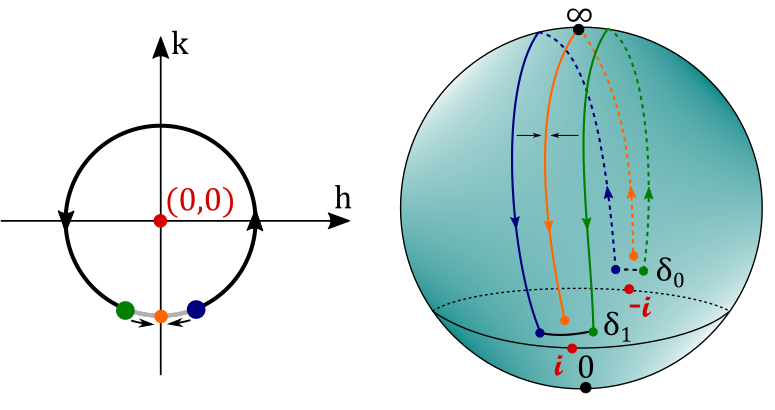}
	\caption{The left panel depicts the loop $\Gamma$ in the space $(h,k)$. The position of the blue, green, yellow and red points are respectively $\rho e^{i(-\frac{\pi}{2}+\epsilon)}$, $\rho e^{i(\frac{3\pi}{2}-\epsilon)}$, $\rho e^{\frac{3i\pi}{2}}$ and $0$. The corresponding trajectories are represented schematically on the Riemann sphere in the right panel with the same color code. The segments $\delta_0$ and $\delta_1$ joining respectively the initial and final points of the trajectories are also plotted (see the text for details).}
	\label{figrisphere}
\end{figure}
Let $\delta_0$ be the segment starting and ending respectively at the initial points of $\tilde{\gamma}_{rel}(\rho e^{i(\frac{-\pi}{2}+\epsilon)})$ and $\tilde{\gamma}_{rel}(\rho e^{i(\frac{3\pi}{2}-\epsilon)})$ and $\delta_1$ the one starting and ending respectively at the final points of $\tilde{\gamma}_{rel}(\rho e^{i(\frac{3\pi}{2}-\epsilon)})$ and $\tilde{\gamma}_{rel}(\rho e^{i(\frac{-\pi}{2}+\epsilon)})$. The circle $\Gamma$ in the space $(h,k)$ and the different paths on the Riemann sphere are represented schematically  in Fig.~\ref{figrisphere}. Numerical results are given in Fig.~\ref{figquasi2}.
\begin{figure}[h]
	\centering
	\includegraphics[width = 0.55\textwidth]{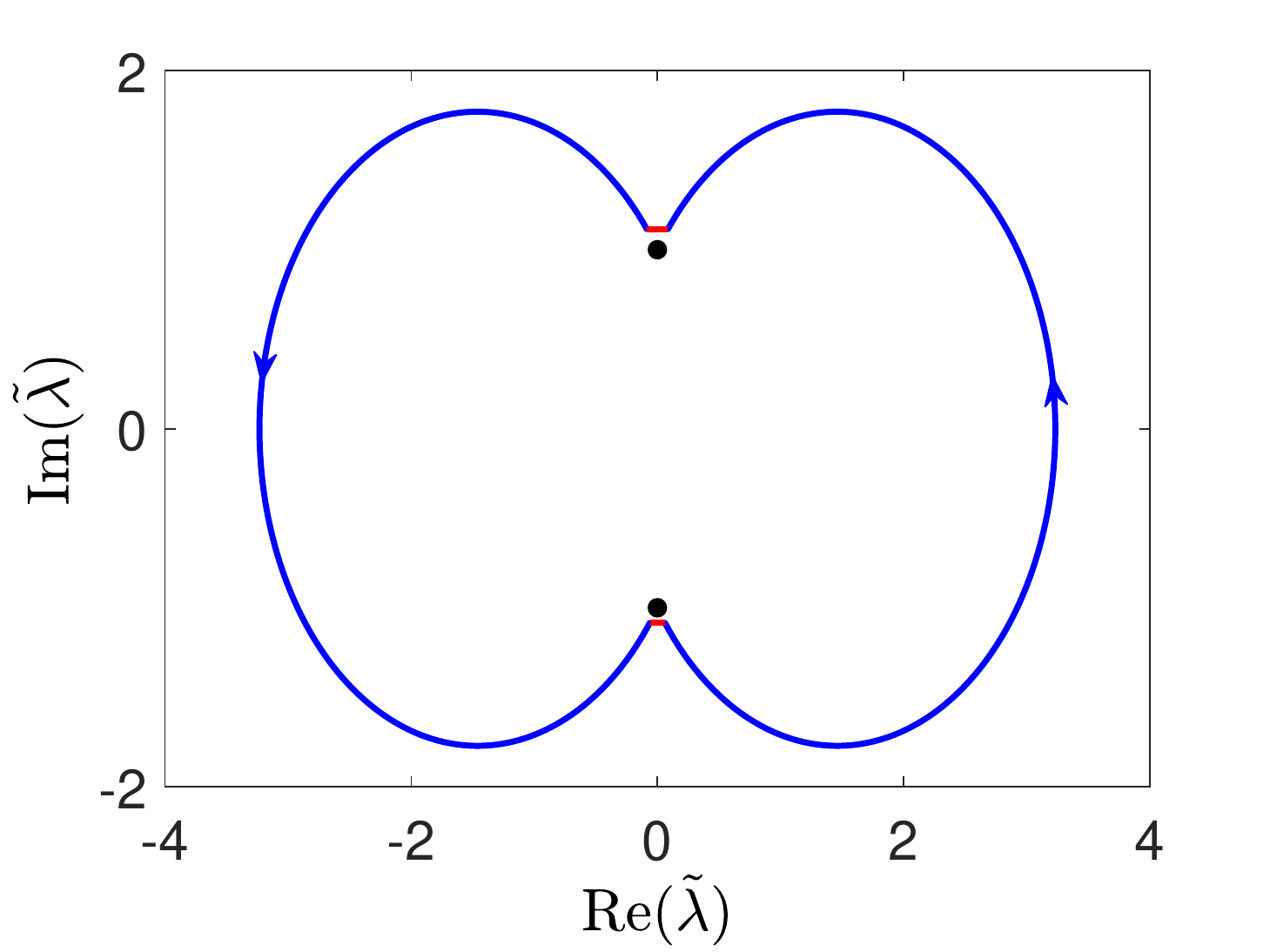}
	\caption{Plot of the cycle $\delta_\epsilon$. The relative cycles and the segments $\delta_0$, $\delta_1$ are respectively represented in blue and red. The black dots correspond to $\pm i$. Numerical values are set to $\rho=0.1$, $a_0=0.4$ and $\epsilon=0.6$. The neighborhood $V$ is a ball of radius 1.}
	\label{figquasi2}
\end{figure}
The segments $\delta_0$ and $\delta_1$ are contained in the domain of analyticity of the function $\log(1+\tilde\lambda^2)$, where $\log$ is taken in the $(0,2\pi)$- branch. This domain is $\CC\setminus\RR\cup i[-1,1]$. We have
\begin{equation*}
\lim_{\epsilon\to 0^+}\left( \int_{\delta_j}\frac{-\tilde\lambda}{i(1+\tilde{\lambda}^2)}d\tilde\lambda\right)=0 ,
\end{equation*}
with $j=0,1$. Using this result, we obtain
\begin{equation*}
\Delta_{\Gamma}\Theta=\lim_{\epsilon\to 0^+}\left( \int_{\delta_\epsilon}\frac{-\tilde\lambda}{i(1+\tilde{\lambda}^2)}d\tilde\lambda\right) ,
\end{equation*}
where $\delta_\epsilon$ is $\tilde{\gamma}_{rel}(\rho e^{i(\frac{3\pi}{2}-\epsilon)})*\delta_1*-\tilde{\gamma}_{rel}(\rho e^{i(\frac{-\pi}{2}+\epsilon)})*\delta_0$, $*$ corresponding to the concatenation of paths.

Finally, we get for $\epsilon>0$
\begin{align*}
\int_{\delta_\epsilon}\frac{-\tilde\lambda}{i(1+\tilde{\lambda}^2)}d\tilde\lambda&=2\pi i\times\textrm{Res}\,\left(\frac{-\tilde{\lambda}}{i(1+\tilde{\lambda}^2)}d\tilde{\lambda},\tilde{\lambda}=\infty\right)\\
&=2\pi ,
\end{align*}
which leads to
\begin{equation*}
\Delta_{\Gamma}\Theta=\lim_{\epsilon\to 0^+}(2\pi)=2\pi.
\end{equation*}

We conclude that the Monodromy matrix of any system with a focus-focus singularity can be computed from this quasi-Lax pair and the corresponding change of variables defined by  $\tilde{\lambda}$ and $\tilde{\mu}$.

\section{Concluding remarks and prospectives}\label{sec.conclu}
Through the analysis of the complex geometry of a Riemann surface given by the Lax pair of a Hamiltonian system, we have shown that Hamiltonian Monodromy can be completely described by the properties of this surface defined by a polynomial of degree four. More precisely, the Monodromy matrix can be computed from a meromorphic form of this surface with a pole at infinity. The general approach is illustrated by two relevant examples, namely the Jaynes-Cumming model and the spherical pendulum for which a Lax pair is known. However, the main weak point of this method is the need to know a Lax pair. This description has been achieved only for few integrable systems, while Hamiltonian Monodromy and its generalizations appear generically in such systems, for instance in any system with a focus-focus singularity. We propose to answer this question by deriving a quasi-Lax pair, i.e. a Lax pair valid (up to higher order terms) in a small neighborhood of the singularity, which allows us to apply with some adaptations the general results of this study. This idea originates from the fact that Hamiltonian Monodromy is a local phenomenon that does not depend on the global Hamiltonian dynamics. A precise mathematical formulation of quasi-Lax pair would be a valuable tool in the study of Hamiltonian singularities.

In addition, this work paves the way for solving other problems in the same framework. A key advantage of Lax pairs is that they allow the description of integrable Hamiltonian systems with a large number of degrees of freedom. An example is given by the Tavis-Cummings system which generalizes the JC model to the case of several spins with different frequencies~\cite{babelon2015}. An extension of our method to such systems may lead to the derivation of high-dimensional Monodromy matrices which characterize their singularities. So far more traditional methods have generally only been used for systems up to three-degree of freedom, and it seems difficult to go further.

\noindent\textbf{Acknowledgments.}\\
This work was supported by the EIPHI Graduate School (contract ANR-17-EURE-0002), the Region and the University of Bourgogne-Franche-Comt\'e. We thank H. R. Jauslin for helpful comments.

\appendix

\section{Hamiltonian Monodromy in the JC model}\label{appA}
We propose in this section another way to compute the expression of the one-form $d\theta$ using the relation between the original coordinates and the new set of variables $(\tilde{\lambda},\tilde{\mu})$. To this aim, we establish the expression of the Liouville two-form $\Omega$ in the new coordinates. We start from
\begin{equation*}
\Omega=idb\wedge d\bar{b}-\frac{i}{S_+}dS_+\wedge dS_z .
\end{equation*}
Using $S_+=-\frac{2}{g}\bar{b}(\tilde{\lambda}-\omega_0)$, we have
\begin{equation*}
\frac{dS_+}{S_+}=\frac{d\bar{b}}{\bar{b}}+\frac{d\tilde{\lambda}}{\tilde{\lambda}-\omega_0} .
\end{equation*}
Since
\begin{equation*}
-id\ln\bar{b}\wedge dK=idb\wedge d\bar{b}-i\frac{d\bar{b}}{\bar{b}}\wedge dS_z,
\end{equation*}
and
\begin{equation*}
d\tilde{\lambda}\wedge d\tilde{\mu} =d\tilde{\lambda}\wedge dS_z,
\end{equation*}
we obtain
\begin{equation*}
\Omega=-id\ln\bar{b}\wedge dK-\frac{i}{\tilde\lambda-\omega_0}d\tilde{\lambda}\wedge d\tilde{\mu},
\end{equation*}
which allows to find the expression of the conjugate angle $\theta$ of $K$ as $\theta=-i\ln\bar{b}$. The Poisson brackets of $\theta$ with $K$, $\tilde{\lambda}$ and $\tilde{\mu}$ can also be directly verified from the expression of $\Omega$, i.e. $\{\theta,K\}=1$ and $\{\theta,\tilde{\lambda}\}=\{\theta,\tilde{\mu}\}=0$.
\section{Hamiltonian Monodromy in the spherical pendulum}\label{appB}
We show in this section how to derive the expression of the rotation number starting from the original coordinates of the phase space.

We first introduce the spherical coordinates $(\vartheta,\varphi)$ defined as
\begin{eqnarray*}
& & x=\sin\vartheta\cos\varphi \\
& & y=\sin\vartheta\sin\varphi \\
& & z=\cos\vartheta
\end{eqnarray*}
Using the corresponding momenta $p_\vartheta$ and $p_\varphi$ such that the only non-zero Poisson brackets are $\{\vartheta,p_\vartheta\}=1=\{\varphi,p_\varphi\}$, we have
\begin{eqnarray*}
& & K_x=-p_\vartheta\sin\varphi-p_\varphi\frac{\cos\varphi}{\tan\vartheta} \\
& & K_y=p_\vartheta\cos\varphi-p_\varphi\frac{\sin\varphi}{\tan\vartheta} \\
& & K_z=p_\varphi
\end{eqnarray*}
Note that straightforward computations lead to
\begin{equation*}
\tilde{\lambda}=\frac{\sin\vartheta}{ip_\vartheta-\frac{p_\varphi}{\tan\vartheta}}
\end{equation*}
The Liouville two-form $\Omega$ can be expressed in these coordinates as follows
\begin{equation*}
\Omega=d\vartheta\wedge dp_\vartheta+d\varphi\wedge dK
\end{equation*}
We consider another two-form given by
\begin{equation*}
\Omega'=\frac{i}{\tilde{\lambda}}d\tilde{\lambda}\wedge d\tilde{\mu}+d\theta\wedge dK
\end{equation*}
where the angle $\theta$ is defined as
\begin{equation}\label{eqthetaSP}
\theta=-i\ln(K_x+iK_y)
\end{equation}
with the properties $\dot{\theta}=\tilde{\lambda}$, $\{\theta,K\}=1$ and $\{\theta,\tilde{\lambda}\}=\{\theta,\tilde{\mu}\}=0$. The angle $\theta$ has all the required properties to be the conjugate angle to $K$ in this new set of coordinates. This result can be explicitly derived by showing that $\Omega=\Omega'$. Starting from $\Omega'$, we have
\begin{equation*}
d\tilde\lambda\wedge d\tilde\mu=d\tilde\lambda\wedge (-\sin\vartheta d\vartheta-\tilde\lambda dK)
\end{equation*}
A straightforward computation leads to
\begin{equation*}
\frac{i}{\tilde\lambda^2}d\tilde\lambda = -\frac{dp_\vartheta}{\sin\vartheta}+i\frac{\cos\vartheta}{\sin^2\vartheta}dK+\frac{i}{\tan\vartheta \tilde\lambda}d\vartheta
-\frac{iK}{\sin^3\vartheta}d\vartheta
\end{equation*}
We deduce that
\begin{equation*}
\frac{i}{\tilde\lambda^2}d\tilde\lambda\wedge d\mu = d\vartheta\wedge dp_\vartheta+\frac{i\tilde\lambda K}{\sin^3\vartheta}d\vartheta\wedge dK
+\frac{\tilde\lambda}{\sin\vartheta}dp_\vartheta\wedge dK
\end{equation*}
We then compute the second term. We have
\begin{equation*}
\theta=-i\ln(\frac{K}{\tan\vartheta}+ip_\vartheta)+\varphi
\end{equation*}
which leads to
\begin{equation*}
d\theta = \frac{-i}{\frac{K}{\tan\vartheta}+ip_\vartheta}(-\frac{K}{\sin^2\vartheta}d\vartheta+idp_\vartheta)+d\varphi
\end{equation*}
and
\begin{equation*}
d\theta\wedge dK=d\varphi\wedge dK-\frac{i\tilde\lambda K}{\sin^3\vartheta}d\vartheta\wedge dK-\frac{\tilde\lambda}{\sin\vartheta}dp_\vartheta\wedge dK
\end{equation*}
We finally get that $\Omega=\Omega'$. We conclude that $\theta$ is given by Eq.~\eqref{eqthetaSP}, which then allows to compute the rotation number $\Theta$.

\bibliographystyle{unsrt}




\end{document}